\documentclass[letter,twocolumn]{autart}

\usepackage{graphicx,subfig}
\usepackage{amsmath,amssymb,amstext,amsfonts,mathtools}
\usepackage{epsfig}
\usepackage{psfrag}
\usepackage{stfloats,color}
\usepackage{fancyhdr}
\usepackage[T1]{fontenc}
\usepackage{nccmath}
\usepackage{algorithm,algorithmic}
\usepackage{wasysym}
\usepackage{accents}
\usepackage{epstopdf}
\usepackage{leftidx}
\usepackage{natbib} 
\usepackage{enumitem}
\usepackage[font=small,labelfont=bf]{caption}
\usepackage{tikz}
\usetikzlibrary{shapes, arrows, decorations, decorations.pathreplacing, decorations.pathmorphing, fit, matrix}
\usepackage[hidelinks]{hyperref}

\newtheorem{theorem}{Theorem}[section]
\newtheorem{proposition}[theorem]{Proposition}

\newtheorem{remark}[theorem]{Remark}
\newtheorem{example}[theorem]{Example}

\newtheorem{assumption}[theorem]{Assumption}
\newtheorem{problem}{Problem}

\input{tcilatex}

\newcommand{\longthmtitle}[1]{\mbox{}\textit{(#1):}}

\newcommand{\real}{\ensuremath{\mathbb{R}}}

\newcommand{\realnonneg}{\ensuremath{\mathbb{R}_{\ge 0}}}

\newcommand{\intnonneg}{{\mathbb{Z}}_{\ge 0}}

\DeclareMathOperator*{\esssup}{\text{ess}\sup}
\newcommand{\map}[3]{#1: #2 \rightarrow #3}
\newcommand{\setdef}[2]{\{#1 \, | \, #2\}}
\newcommand{\setdefb}[2]{\big\{#1 \; | \; #2\big\}}

\newcommand{\Gc}{\mathcal{G}}

\newcommand{\Kc}{\mathcal{K}}
\newcommand{\Lc}{\mathcal{L}}

\newcommand{\Nc}{\mathcal{N}}

\newcommand{\var}{\text{var}}

\renewcommand{\theenumi}{(\roman{enumi})}

\newcommand{\oprocendsymbol}{\hbox{$\bullet$}}
\newcommand{\oprocend}{\relax\ifmmode\else\unskip\hfill\fi\oprocendsymbol}

\newcommand\new[1]{{\color{blue} #1}}
\renewcommand\new[1]{{#1}}
\newcommand\nnew[1]{{\color{blue} #1}}
\renewcommand\nnew[1]{{#1}}

\graphicspath{{figures/}}
    
\begin{document}

\begin{frontmatter}

 \title{Event-Triggered Stabilization of Nonlinear Systems with Time-Varying Sensing and Actuation Delay}

 \thanks{A preliminary version of this paper appeared at the IEEE
   Conference on Decision and Control as~\citep{EN-PT-JC:16-cdc}.}
    
    \vspace{-10pt}
  
  \author[First]{Erfan Nozari}%
  \author[Second]{\quad Pavankumar Tallapragada}%
  \author[Third]{\quad Jorge Cort\'es}
  
    \address[First]{Department of Electrical and Systems Engineering, University of Pennsylvania, enozari@seas.upenn.edu}
    \address[Second]{Department of Electrical
    Engineering, Indian Institute of Science, pavant@ee.iisc.ernet.in}
    \address[Third]{Department of Mechanical and Aerospace Engineering,
    University of California, San Diego,
   cortes@ucsd.edu}
  
\begin{abstract}
  This paper studies the problem of stabilization of a nonlinear
  system with time-varying delays in both sensing and actuation using
  event-triggered control. Our proposed strategy seeks to
  opportunistically minimize the number of control updates while
    guaranteeing stabilization and builds on predictor feedback to
  compensate for arbitrarily large known time-varying delays. We
  establish, using a Lyapunov approach, the global asymptotic
  stability of the closed-loop system as long as the open-loop system
  is globally input-to-state stabilizable in the absence of time
  delays and \nnew{sampling}. We further prove that the proposed
  event-triggered law has inter-event times that are uniformly lower
  bounded and hence does not exhibit Zeno behavior. For the particular
  case of a stabilizable linear system, we show global exponential
  stability of the closed-loop system and analyze the trade-off
  between the rate of exponential convergence and \nnew{a bound on the
  sampling frequency}. We illustrate these results in simulation
  and also examine the properties of the proposed event-triggered
  strategy beyond the class of systems for which stabilization can be
  guaranteed.
\end{abstract}
  
  
\end{frontmatter}

\section{Introduction}\label{sec:intro}

Event- and self-triggered approaches have recently gained popularity
for controlling cyberphysical systems. The basic premise is that of
abandoning the assumption of continuous or periodic updating of the
control signal and instead adopt an opportunistic perspective that
leads to deliberate, aperiodic updates. The challenge resides in
determining precisely when control signals should be updated to
improve efficiency while still guaranteeing convergence.  This paper
expands the state-of-the-art in resource-aware control by designing
predictor-based event-triggered control strategies that stabilize
nonlinear systems with \emph{known} delays in both sensing and
actuation that can be \emph{arbitrarily large} and
\emph{time-varying}.

\emph{Literature review:} There exists a vast literature on both
event-triggered control and the control of time-delay systems. Here,
we review the works most closely related to our treatment. Originating
from event-based and discrete-event
systems~\citep{CGC-SL:07,LZ-ZDW-DHZ:17},
the concept of event-triggered control (i.e., updating the control
signal in an opportunistic fashion) was proposed
in~\citep{HK:91,KJA-BMB:02} and has found its way into the efficient
use of sensing, computing, actuation, and communication resources in
networked control systems, see
  e.g.,~\citep{PT:07,XW-MDL:11,WPMHH-KHJ-PT:12,MA-RP-JD-DN:17} and
references therein.
On the other hand, the notion of predictor feedback
is a powerful method in dealing with controlled systems subject to
time delay~\citep{OJMS:59,DQM:68,AM-AO:79,MTN:91,MK:09,IK-MK:12}. In
essence, a predictor feedback controller anticipates the future
evolution of the plant using its forward model and sends the control
signal early enough to compensate for the delay.  Here, we pursue a
Lyapunov-based analysis of predictor feedback
following~\citep{NBL-MK:13}.  Given that numerical implementations of
predictor feedback controllers are particularly
challenging~\citep{LM:04,QCZ:04}, we further discuss several methods
for the implementation of our proposed controller and show that a
carefully designed ``closed-loop'' method is numerically stable and
robust to errors in delay compensation.

The joint treatment of time delay and event-triggering is particularly
challenging. 
By its opportunistic nature, an event-triggered controller keeps the
control value unchanged until the plant is close to instability and
then updates the control value according to the current state. Now, if
time delays exist, the controller only has access to some past state
of the plant (delayed sensing) and it takes some time for an updated
control action to reach the plant (delayed actuation), jointly
increasing the possibility of the updated control value being already
obsolete when it is implemented in the plant, resulting in
instability. Therefore, the controller needs to be sufficiently
proactive and update the control value sufficiently ahead of time to
maintain closed-loop stability. This makes the design problem
challenging.  Delays in actuation and sensing may be due to
  communication delays between controller-actuator and
  controller-sensor pairs, and in that sense, previous work on the
  event-triggered control literature that specifically considers
  delays in the communication channel deals with a similar problem
  setup as the one considered here. Several event-triggered designs
  consider scenarios where the system dynamics are linear, see,
  \new{e.g.~\citep{XMZ-QLH-BLZ:17,JC-SM-JS:17,AS-EF:16-aut,AS-EF:16-aut2,XG-QLH:15,EG-PJA:13}.} The
  inclusion of nonlinearity, however, makes the problem more
  challenging.
When digital controllers are used and the delay is smaller than the
sampling time,~\citep{LH-JD-CI:06,WW-SR-DG-SL:15} design
event-triggered controllers for the resulting delay-free discretized
system.
Robust event-triggered stabilizing controllers are also designed
  for nonlinear systems with sensing delays in~\citep{LL-XW-MDL:12a}
  and with both sensing and actuation delays
  \new{in~\citep{PT:07,VSD-DPB-WPMHH:17}}. In all these
  works, 
  however, a key assumption is that the (maximum) delay is smaller
  than the (minimum) inter-transmission time. This assumption (also
  called the small-delay case) allows for the \emph{treatment of delay
    as a disturbance} and, by construction, can tolerate unknown
  delays. In reality, however, (minimum) inter-transmission times can
  be very small, making this assumption restrictive.  Similar to
    our preliminary work~\citep{EN-PT-JC:16-cdc}, we take a different
  perspective here and consider arbitrarily large delays, with the
  expected tradeoff in our treatment that the delay can no longer be
  unknown. The technical approach is based on using predictors that
  capture the effect of the delay on the system to compensate for
  it. We rigorously analyze the case when the delay is accurately
  known and show in simulation that our design is indeed robust to
  small variations when the delay is only approximately known.
  \nnew{Unlike~\citep{EN-PT-JC:16-cdc}, here we consider 
  event-triggering and time-varying delay
    both in sensing and actuation.} Further, given the well-known
    difficulties in the computation of predictor-feedback controllers,
    we here provide a detailed discussion of the numerical challenges
    that arise in the implementation of predictor feedback and
    effective solutions to resolve them. Finally, this paper provides
    a complete and thorough technical treatment, including the proofs
    of all results, which are not available
    in~\citep{EN-PT-JC:16-cdc}.

\emph{Statement of contributions:} Our contributions are
threefold. First, we design an event-triggered controller for
stabilization of
nonlinear systems with arbitrarily large sensing and actuation
delays. We employ the method of predictor feedback to compensate for
the delay in both and then co-design the control law and triggering
strategy to guarantee the monotonic decay of a Lyapunov-Krasovskii
functional. Our second contribution involves the closed-loop analysis
of the event-triggered law, proving that the closed-loop system is
globally asymptotically stable and the inter-event times are uniformly
lower bounded (and thus no Zeno behavior may exist). Due to the
importance of linear systems in numerous applications, we briefly
discuss the simplifications of the design and analysis in this
case. Our final contribution pertains to the trade-off between
convergence rate and sampling. Our analysis in this part is
limited to linear systems, where closed-form solutions are derivable
for (exponential) convergence rate and minimum inter-event
times. We provide a quantitative account of the well-known
  trade-off between sampling and convergence in event-triggered
  designs and show how this trade-off can be biased in either
  direction by tuning a design parameter. Finally, we present
simulations to illustrate the effectiveness of our design and
  address its numerical implementation.

\section{Preliminaries}\label{sec:prelims}

We introduce notational conventions and briefly review notions on
input-to-state stability.  We denote by $\real$ and $\realnonneg$ the
sets of reals and nonnegative reals, respectively. 
\new{Given $t \in \real$ and a function $f$ on $\real$, $t_+ \triangleq \max\{t, 0\}$ while $f(t^+) \triangleq \lim_{s \to t^+} f(s)$ and $f(t^-) \triangleq \lim_{s \to t^-} f(s)$ \nnew{when these limits exist}.}
Given a vector or
matrix, we use $|\cdot|$ to denote the Euclidean norm.  We denote by
$\Kc$ the set of strictly increasing continuous functions $\alpha: [0,
\infty) \to [0, \infty)$ with $\alpha(0) = 0$. $\alpha$ belongs to
$\Kc_\infty$ if $\alpha \in \Kc$ and $\lim_{r \to \infty} \alpha(r) =
\infty$. We denote by $\Kc \Lc$ the set of \new{continuous} functions $\beta:[0,
\infty) \times [0, \infty) \to [0, \infty)$ such that, for each $s \in
[0, \infty)$, $r \mapsto \beta(r,s)$ \new{belongs to class $\Kc$} 
and, for each $r \in [0, \infty)$, $s \mapsto
\beta(r,s)$ is monotonically decreasing with $\beta(r, s) \to 0$ as $s
\to \infty$.  We use the notation $\Lc_f S = \nabla S \cdot f$ for the
Lie derivative of a function $\map{S}{\real^n}{\real}$ along the
trajectories of a vector field~$f$ taking values in~$\real^n$.

We follow~\citep{EDS-YW:95} to review the definition of input-to-state
stability of nonlinear systems and its Lyapunov characterization.
Consider a nonlinear system of the form
\begin{align}\label{eq:iss-system}
  \nnew{\dot x(t) = f(x(t), u(t)), \qquad \text{a.a. } t \ge 0, \qquad x(0) = x_0,}
\end{align}
where $f: \real^n \times \real^m \to \real^n$ is continuously
differentiable, $f(0, 0) = 0$, and \nnew{``a.a." (almost all)}
  denotes the fact that $x$ may not be differentiable on a set of
  Lebesgue measure zero.
System~\eqref{eq:iss-system} is (globally) input-to-state
stable (ISS) if there exist $\alpha \in \Kc$ and $\beta \in \Kc\Lc$
such that for any measurable locally essentially bounded input
$u:\realnonneg \to \real^m$ and any initial condition $x(0) \in
\real^n$, its solution satisfies
\begin{align*}
  |x(t)| \le \beta(|x(0)|, t) + \alpha\big(\esssup\nolimits_{t \ge 0}
  |u(t)|\big),
\end{align*}
for all $t \ge 0$.  For this system, a continuously differentiable
function $S:\real^n \to \realnonneg$ is called an ISS-Lyapunov
function if
there exist $\alpha_1, \alpha_2, \gamma,
\rho \in \Kc_\infty$ such that
\begin{subequations}\label{eq:iss-lyap2}
  \begin{align}
    &\forall x \in \real^n && \alpha_1(|x|) \le S(x) \le
    \alpha_2(|x|),
    \\
  \label{eq:iss-lyap2a}  &\forall (x, u) \in \real^{n + m} && \Lc_f S(x, u) \le
    -\gamma(|x|) + \rho(|u|).
  \end{align}
\end{subequations}
According to~\citep[Theorem 1]{EDS-YW:95}, the
system~\eqref{eq:iss-system} is ISS if and only if it admits an
ISS-Lyapunov function.

\section{Problem Statement}\label{sec:prob-state}

Consider the nonlinear 
system (``plant'') with dynamics
\begin{align}\label{eq:dynamics}
  \nnew{\!\!\dot x(t) = f(x(t), u_p(t)), \qquad \text{a.a. } t \ge 0, \qquad x(0) = x_0,}
\end{align}
where $f:\real^n \times \real^m \to \real^n$. 
Our goal is to provide a state-feedback controller ensuring global
asymptotic stability under the following 
challenges:
\begin{enumerate}[wide]
\item \textbf{Actuation delay:} Let $u(t)$ be the control signal
  generated by the controller. Actuation delay is modeled as
  \begin{align}
    u_p(t) = u(\phi(t)), \quad t \ge 0,
  \end{align}
  where $t - \phi(t) > 0$ is the amount of time that it takes
  for a control action generated at time $\phi(t)$ to reach the
  plant/actuator.
  For instance, In the case of a constant actuation delay $D$, we have
  $\phi(t) = t - D$. \new{This delay further requires an initial value
    $\{u(t) \;|\; \phi(0) \le t < 0\}$ on the control input
    for~\eqref{eq:dynamics} to be well-defined.}
\item \textbf{Sensing delay:} \new{We allow the existence of a delay
    between the sensor and the controller such that at any time $t$,
    the controller may have access to $x(s), s \le \psi(t)$
    (alternatively, $x(t)$ takes $\psi^{-1}(t) - t$ seconds to reach
    the controller) for some delay function $\psi(t) \le t$.}
\item \textbf{Actuation event-triggering:} 
    We seek to design a controller that updates $u(t)$
  only at a sequence of discrete times $\{t_k\}_{k = 0}^\infty$,
  \begin{align}\label{eq:u1}
    u(t) = u(t_k), \quad t \in [t_k, t_{k + 1}), \quad k \ge 0.
  \end{align}
  \item \textbf{Sensing event-triggering:} \nnew{We further allow for
    the possibility that the event-triggering mechanism does not have
    access to the plant state at all times $t \in \realnonneg$, but
    only at some time instants denoted $\tau_\ell, \ell \in
    \intnonneg$.}%
  \footnote{\nnew{We note that this sampled sensing scheme is also
      called \emph{periodic event-triggered control}, even when the
      sampling times are not equally spaced. Nevertheless, we do not
      adopt this terminology here to avoid the latter
      interpretation.}}
  \new{In this case, we let for simplicity that $\tau_0 = 0$,
  $t_0 = \psi^{-1}(0)$, and $u(t)$ be arbitrarily set in $[0,
    t_0)$ as the controller has not received any state information
    yet.}
\end{enumerate}

In the sequel, we impose the following assumptions on the system dynamics.

\begin{assumption}\longthmtitle{Standing assumptions}\label{assum}
  \rm \begin{enumerate}
  \item\label{item:assum-f} $f$ is continuously differentiable, $f(0,
    0) = 0$, and~\eqref{eq:dynamics} is forward complete (does not
    exhibit finite escape time) for all initial conditions and bounded
    inputs;
  \item the initial control $\{u(t) \;|\; \phi(0) \le t < 0\}$ is
    given and continuously differentiable;
  \item\label{item:assum-c1} the delay function $\phi$ is continuously
    differentiable;
  \item\label{item:assum-time} the delay functions $\phi$ and $\psi$
    are monotonically increasing so the argument of $u(\phi(t))$ and
    $x(\psi(t))$ do not go back in time;
  \item\label{item:assum-iss} the origin of~\eqref{eq:dynamics} is
    robustly globally asymptotically stabilizable in the absence of
    delays and with continuous sensing and actuation. Formally, there
    exists a globally Lipschitz feedback law $K:\real^n \to \real^m$,
    $K(0) = 0$, that makes 
    \begin{align}\label{eq:iss-assum}
      \dot x(t) = f(x(t), K(x(t)) + w(t)),
    \end{align}
    ISS with respect to the additive input disturbance~$w$;
  \item\label{item:assum-phipsi} the delay function $\phi$ is known to
    the controller; on the other hand, $\psi$ need not be known a
    priori or \nnew{for all times}, but only a posteriori and at times when state
    is measured;
  \item\label{item:assum-bdd} the delay function $\phi$ and its
    derivative are bounded, i.e., there exist $M_0 > 0$, $M_1 \ge 1$,
    and $0 < m_2 \le 1$ such that
    \begin{align}\label{eq:bounds-delay}
      \hspace{-5pt} t - \phi(t) \le M_0 \ \ \text{and} \ \ m_2 \le
      \dot \phi(t) \le M_1, \ \ \forall t \ge 0;
    \end{align}
  \item\label{item:assum-psi} the sensing triggering times
    $\{\tau_\ell\}_{\ell = 0}^\infty$ are given (determined by the
    sensor independently of our design). In particular, the sensor
    ensures that $\{\nnew{\tau_\ell}\}_{\ell \ge 0} \cap [a, b]$ is finite for
    any $a, b < \infty$ (lack of Zeno
    behavior) \nnew{while $\{\tau_\ell\}_{\ell = 0}^\infty$ can be arbitrary otherwise}. 
    \oprocend
  \end{enumerate}
\end{assumption}

Assumption~\ref{assum}\ref{item:assum-f}-\ref{item:assum-time}
  are standard in predictor-based control of delay systems. 
  \nnew{In the case of digital communications, Assumption~\ref{assum}\ref{item:assum-time} requires the lack of packet reordering. Nevertheless, the nature of the control system is such that any $u(t_k)$ has become obsolete and can be safely discarded, should it arrive later than $u(t_i), i \ge k$. The same applies to $\{x(\tau_{\ell})\}_{\ell = 1}^\infty$. Thus, $\phi$ and $\psi$ can be, without loss of generality, replaced by a monotonically increasing upper bound if they are not so originally.}
 Assumption~\ref{assum}\ref{item:assum-f}, together
  with the piecewise-constant form of $u_p$,
   further ensures existence and
  uniqueness of solutions for~\eqref{eq:dynamics}.
  Assumption~\ref{assum}\ref{item:assum-iss} 
  is also
  standard in event-triggered control, \nnew{though not necessarily with the globally Lipschitz property assumed here}. This allows us to focus on the
  challenges that arise by time delays and event-triggered control.
  Further, the a priori knowledge of $\phi$ in
  Assumption~\ref{assum}\ref{item:assum-phipsi} is most realistic in
  applications where the same control task is repeatedly executed and
  thus a data-driven estimate of future $\phi$ can be computed using
  its history. Moreover, note that
  Assumption~\ref{assum}\ref{item:assum-bdd} is trivially satisfied
  for a constant delay ($\phi(t) = t - D$) with $M_0 = D$ and $M_1 =
  m_2 = 1$. Finally, Assumption~\ref{assum}\ref{item:assum-psi} is
  imposed for simplicity and to let us focus on the design of the
  actuation triggering times. In fact, the values of $\{\tau_\ell\}$
  other than $\tau_0$ are irrelevant theoretically but practically
  critical for stability, a point we discuss in detail in
  Sections~\ref{sec:sensing} and~\ref{sec:sims}.

\new{The resulting networked control scheme is illustrated in
  Figure~\ref{fig:struct}. Our considered problem is then as follows.}

\begin{problem}\longthmtitle{Event-Triggered Stabilization under
    Sensing and Actuation Delay}\label{problem}
  Design the sequence of actuation triggering times%
  \footnote{Recall that $t_0 = \psi^{-1}(0)$ is fixed.}  $\{t_k\}_{k =
    1}^\infty$ and the corresponding control values $\{u(t_k)\}_{k =
    0}^\infty$ such
  that 
  \nnew{$\{t_{k + 1} - t_k\}_{k \ge 0}$ is uniformly lower bounded by
    a strictly positive constant} and the closed-loop
  system~\eqref{eq:dynamics} is globally asymptotically stable using
  the piecewise constant control~\eqref{eq:u1} and the
  delayed information $\{x(\tau_\ell)\}_{\ell = 0}^\infty$ received,
  resp., at $\{\psi^{-1}(\tau_\ell)\}_{\ell = 0}^\infty$.%
  \footnote{
    We require that the control law is causal, i.e., $t_k$ and
    $u(t_k)$ depend only on the states $\{x(\tau_\ell)\}$ that have
    reached the controller by the time $t_k$. While sampling may be
    modeled as a specific type of delay, we capture it with the
    prediction error $e(t)$ (defined later). The values $\phi(t)$ and
    $\psi(t)$ only capture the delays in actuation and sensing, resp.}
\oprocend
\end{problem}

The requirement that \new{$\{t_k\}_{k \ge 0} \cap [a, b]$ be finite
  for any $0 \le a \le b <
  \infty$} 
ensures the resulting design is implementable by avoiding finite
accumulation points, i.e., Zeno behavior. We propose a solution
  to Problem~\ref{problem} in the next section.

\begin{figure}
\begin{tikzpicture}[>=latex, every node/.style={draw, line width=1pt, rounded corners}, every path/.style={line width=1pt}, node distance=40pt]
  \node[rectangle, inner sep=6pt] (plant) {Plant}; %
  \node[left of=plant, isosceles triangle, inner sep=2pt, xshift=-30pt] (zoh) {\small{ZOH}}; %
  \draw[->, shorten <=-3pt] (zoh) to node[xshift=-5pt, yshift=8pt, draw=none]{\nnew{${\scriptstyle u_p(t)}$}} (plant); %
  \node[right of=plant, circle, line width=0.1pt, inner sep=0pt, xshift=10pt] (bend-midExt) {}; %
  \draw[shorten >=-1pt] (plant) to (bend-midExt); %
  \node[above right of=bend-midExt, circle, inner sep=0.7pt, fill, xshift=-15pt, yshift=-15pt] (openExt) {}; %
  \draw[shorten <=-0.3pt] (bend-midExt) to (openExt); %
  \node[right of=bend-midExt, circle, inner sep=0.7pt, fill, xshift=-20pt] (closedExt) {}; %
  \node[left of=openExt, circle, line width=0.1pt, inner sep=0pt, xshift=28pt] (arr-tailExt) {}; %
  \node[coordinate, left of=closedExt, line width=0.1pt, inner sep=0pt, xshift=34pt, yshift=-8pt] (arr-headExt) {}; %
  \draw[->, shorten <=-1pt, shorten >=-2pt, bend left=20] (arr-tailExt) to (arr-headExt); %
  \node[draw=none, left of=arr-tailExt, xshift=33pt, yshift=4pt] (tauk) {${\scriptstyle \tau_\ell}$}; %
  \node[right of=plant, circle, line width=0.1pt, inner sep=0pt, xshift=55pt] (bend-top-right) {}; %
  \draw[shorten >=-1pt] (closedExt) to (bend-top-right); %
  \node[below right of=plant, circle, xshift=67pt, yshift=-6pt] (netExt) {${\scriptstyle \psi(t)}$}; %
  \draw[shorten <=-1pt, shorten >=-1pt] (bend-top-right) to (netExt); %
  \node[below of=plant, rectangle, inner sep=6pt, xshift=0pt, yshift=-30pt] (controller) {Controller}; %
  \node[left of=controller, circle, line width=0.1pt, inner sep=0pt, xshift=-20pt] (bend-mid) {}; %
  \draw[shorten >=-1pt] (controller) to (bend-mid); %
  \node[above left of=bend-mid, circle, inner sep=0.7pt, fill, xshift=15pt, yshift=-15pt] (open) {}; %
  \draw[shorten <=-0.3pt] (bend-mid) to (open); %
  \node[left of=bend-mid, circle, inner sep=0.7pt, fill, xshift=20pt] (closed) {}; %
  \node[right of=controller, circle, line width=0.1pt, inner sep=0pt, xshift=55pt] (bend-bottom-right) {}; %
  \draw[shorten <=-1pt, shorten >=-1pt] (netExt) to (bend-bottom-right); %
  \draw[->, shorten <=-1pt] (bend-bottom-right) to (controller); %
  \node[below left of=plant, circle, xshift=-80pt, yshift=-6pt] (net) {${\scriptstyle \phi(t)}$}; %
  \node[left of=closed, circle, line width=0.1pt, inner sep=0pt, xshift=12pt] (bend-bottom-left) {}; %
  \draw[shorten >=-1pt] (closed) to (bend-bottom-left); %
  \draw[->, shorten <=-1pt] (bend-bottom-left) to (net); %
  \node[left of=zoh, circle, line width=0.1pt, inner sep=0pt, xshift=1pt] (bend-top-left) {}; %
  \draw[shorten >=-1pt] (net) to (bend-top-left); %
  \draw[->, shorten <=-1pt] (bend-top-left) to (zoh); %
  \node[right of=open, circle, line width=0.1pt, inner sep=0pt, xshift=-28pt] (arr-tail) {}; %
  \node[coordinate, right of=closed, line width=0.1pt, inner sep=0pt, xshift=-34pt, yshift=-8pt] (arr-head) {}; %
  \draw[->, shorten <=-1pt, shorten >=-2pt, bend right=20] (arr-tail) to (arr-head); 
  \node[draw=none, above right of=plant, xshift=3pt, yshift=-20pt] (xt) {${\scriptstyle x(t)}$}; %
  \node[draw=none, left of=open, xshift=20pt, yshift=-5pt] (utk) {${\scriptstyle u(t_k)}$}; \node[draw=none, left of=zoh, xshift=10pt, yshift=8pt] (uphitk) {${\scriptstyle u(\phi(t_k))}$}; %
  \node[draw=none, right of=arr-tail, xshift=-33pt, yshift=4pt] (tk) {${\scriptstyle t_k}$}; %
  \node[draw=none, right of=closedExt, xshift=-25pt, yshift=8pt] (xtauk) {${\scriptstyle x(\tau_\ell)}$}; %
  \node[draw=none, right of=controller, xshift=20pt, yshift=8pt] (xpsitauk) {${\scriptstyle x(\psi(\tau_\ell))}$}; %
\node[below of=controller, dashed, draw=black!70, minimum width=130pt, minimum height=40pt, yshift=-10pt](controller-box) {};%
\node[right of=controller-box, inner sep=7pt, xshift=-12pt] (predictor) {Predictor}; %
\node[left of=predictor, inner sep=5pt, xshift=-30pt] (K) {$K(\cdot)$}; %
\node[right of=predictor, circle, line width=0.1pt, inner sep=0pt, xshift=20pt] (enter) {}; %
\node[coordinate, left of=K, line width=0.1pt, inner sep=0pt, xshift=-22pt] (exit) {};%
\draw[->, shorten <=-1pt] (enter) to (predictor);%
\draw[->] (predictor) to (K);%
\draw[->, shorten >=-1pt] (K) to (exit);%
\node[draw=none, right of=predictor, xshift=18pt, yshift=8pt] (xt2) {${\scriptstyle x(\psi(\tau_\ell))}$};%
\node[draw=none, left of=predictor, xshift=1pt, yshift=8pt] (pt) {${\scriptstyle p(t)}$};%
\node[draw=none, left of=K, xshift=-1pt, yshift=8pt] (Kpt) {${\scriptstyle K(p(t))}$};%
\draw[line width=0.5pt, shorten >=5pt] (controller.south east) to (controller-box.16);%
\draw[line width=0.5pt, shorten >=3pt] (controller.south west) to (controller-box.164);%
\end{tikzpicture} 
\caption{The considered networked control scheme with sensing and
  actuation delays and event-triggering (top) and the proposed
  predictor-based controller (bottom).}
\label{fig:struct}
\end{figure}

\section{Event-Triggered Design and Analysis}\label{sec:control}

In this section, we propose an event-triggered control policy to solve
Problem~\ref{problem}. We start our analysis with the simpler case
where the controller receives state feedback continuously (i.e.,
$\{x(t)\}_{t = 0}^\infty$ instead of $\{x(\tau_\ell)\}_{\ell =
  0}^\infty$) without delays (i.e., $\psi(t) = t$), and later extend
it to the general case.

\subsection{Predictor Feedback Control for Time-Delay
  Systems}\label{subsec:pf}

Here we review the continuous-time stabilization of the
dynamics~\eqref{eq:dynamics} by means of a predictor-based feedback
control~\citep{NBL-MK:13}. For convenience, we denote the inverse of
$\phi$ by
  $\sigma(t) = \phi^{-1}(t)$,
for all $t \ge 0$. The inverse exists since $\phi$ is strictly
monotonically increasing. From~\eqref{eq:bounds-delay}, for
all $t \ge \phi(0)$,
\begin{align*}
  \dot \sigma(t) \le M_2 \triangleq m_2^{-1}.
\end{align*}
To compensate for the delay, at any time $t \ge
\phi(0)$, the controller makes the following prediction of the future
state of the plant,
\begin{align}\label{eq:p}
  p(t) = x(\sigma(t)) = \new{x(t_+) + \int_{\phi(t_+)}^t} \!\! \dot \sigma(s)
  f(p(s), u(s)) d s.
\end{align}
This integral is computable by the
controller since it only requires knowledge of the initial or current
state of the plant and the history of
$u(t)$ and $p(t)$, all of which are available to the
controller. 
\nnew{In the remainder, we thus assume that $p(t)$ can be computed exactly, but hint that numerical integration errors can lead to instability if not treated properly. We will give a detailed empirical discussion of this matter in Section~\ref{sec:sims} but its rigorous analysis remains open for future research.}

As shown in Figure~\ref{fig:struct}, the controller
applies the control law $K$ on the prediction $p$ to
compensate for the delay, 
\begin{align}\label{eq:non-et-u}
  u(t) = K(p(t)), \qquad t \ge 0.
\end{align}
The next result shows convergence for the closed-loop system.

\begin{proposition}\longthmtitle{Asymptotic Stabilization by Predictor
    Feedback~\citep{NBL-MK:13}} 
  Under \new{Assumption~\ref{assum}}, the closed-loop
  system~\eqref{eq:dynamics} under the controller~\eqref{eq:non-et-u}
  is globally asymptotically stable, i.e., there exists $\beta \in
  \Kc\Lc$ such that for any $x(0) \in \real^n$ and bounded
  $\{u(t)\}_{t = \phi(0)}^0$, for all $t \ge 0$,
  \begin{align*}
    |x(t)| + \sup_{\phi(t) \le \tau \le t} |u(\tau)| \le
    \beta\Big(|x(0)| + \sup_{\phi(0) \le \tau \le 0} |u(\tau)|,
    t\Big) .
  \end{align*}
\end{proposition}

\subsection{Design of Event-triggered Control
  Law}\label{subsec:event-design}

Following Section~\ref{subsec:pf}, we let the controller make the
prediction $p(t)$ according to~\eqref{eq:p} for all $t \ge
\phi(0)$. Since the controller can only update $u(t)$ at discrete
times $\{t_k\}_{k = 0}^\infty$, it uses the piecewise-constant
control~\eqref{eq:u1} and assigns the control
\begin{align}\label{eq:u2}
  u(t_k) = K(p(t_k)),
\end{align}
for all $k \ge 0$. In order to design the triggering times $\{t_k\}_{k
  = 1}^\infty$, we use Lyapunov stability tools to determine when the
controller has to update $u(t)$ to prevent instability. We define the
triggering error for all $t \ge \phi(0)$ as
\begin{align}\label{eq:e}
  e(t) = \begin{cases} p(t_k) - p(t) \; &\text{if} \; t \in
    [t_k, t_{k+1}) \text{ for  } k \ge 0,
    \\
    0 \; &\text{if} \; t \in [\phi(0), t_0),
  \end{cases}
\end{align}
so that $ u(t) = K(p(t) + e(t))$, for $t \ge t_0$.  Let
\begin{align}\label{eq:w}
  w(t) = u(t) - K (p(t) + e(t)), \qquad t \ge \phi(0),
\end{align}
where $w(t) = 0$ for $t \ge t_0$ but $w(t)$ is in general nonzero for
$t \in [\phi(0), t_0)$. Computing $u(\phi(t))$ from~\eqref{eq:w}
  and substituting it in~\eqref{eq:dynamics}, the closed-loop system
can be written
\begin{align}\label{eq:cl}
  \dot x(t) = f\big(x(t), K\big(x(t) + e(\phi(t))\big) +
  w(\phi(t))\big),
\end{align}
for all $t \ge 0$. Notice that~\eqref{eq:cl} simplifies
to~\citep[Eq.~(3)]{PT:07} in the absence of delay ($\phi(t) = t$).
Let $g(x, w) = f(x, K(x) + w)$ for all $x, w$. By
\new{Assumption~\ref{assum}\ref{item:assum-iss},} 
there exists a continuously differentiable function $S: \real^n \to
\real$ and $\alpha_1, \alpha_2, \gamma, \rho \in \Kc_\infty$ such that
\begin{align}\label{eq:alpha12}
  \alpha_1(|x(t)|) \le S(x(t)) \le \alpha_2(|x(t)|),
\end{align}
and $(\Lc_g S)(x, w) \le \new{-} \gamma(|x|) + \rho(|w|)$. Therefore, we have
\begin{align}\label{eq:Lf_S}
  &(\Lc_f S)\big(x(t), K\big(x(t) + e(\phi(t))\big) + w(\phi(t))\big)
  \\
  \notag &= (\Lc_g S)\big(x(t), K\big(x(t) \!+\!
  e(\phi(t))\big) \!+\! w(\phi(t)) \!-\! K(x(t))\big)
  \\
  \notag &\!\le \!-\gamma(|x(t)|) + \rho\big(\big|K\big(x(t) \!+\!
  e(\phi(t))\big) \!+\! w(\phi(t)) \!-\! K(x(t))\big|\big).
\end{align}
\nnew{As in~\citep[eq.~(8.47)]{NBL-MK:13},} let%
\footnote{\nnew{Note that $\rho$ can always be chosen such that~\eqref{eq:V} is well-defined, e.g., by choosing it such that $\rho(r) / r \in \Kc_\infty$ using~\cite[Thm 1]{ES-AT:95}.}}
\begin{subequations}
  \begin{align}
    \label{eq:V} V(t) &= S(x(t)) + \frac{2}{b} \int_0^{2L(t)} \frac{\rho({r})}{{r}} d
    r,
    \\
    \label{eq:L} L(t) &= \sup_{t \le \tau \le \sigma(t)} |e^{b (\tau - t)}
    w(\phi(\tau))|,  
\end{align}
\end{subequations}
where $b > 0$ is a design parameter. 
\new{Note that the second term in~\eqref{eq:V} may only be nonzero for $t \in [\phi(0), t_0)$ since the system is open-loop over this interval (cf.~\eqref{eq:e},\eqref{eq:w}).} 
The next result establishes an
upper bound on $d V / d t$.

\begin{proposition}\longthmtitle{Upper-bounding $\dot
    V(t)$}\label{prop:ub-V-dot}
  For the system~\eqref{eq:dynamics} under the control defined
  by~\eqref{eq:u1} and~\eqref{eq:u2} and the predictor~\eqref{eq:p},
  we have \nnew{for any solution with maximal interval of existence
    $[0,t_{\max})$,}
  \begin{align}\label{eq:Vdot-bound}
    \dot V(t) \le -\gamma(|x(t)|) - \rho(2 L(t)) + \rho(2 L_K
    |e(\phi(t))|),
  \end{align}
  \nnew{for all $t \in [0, t_{\max}) \setminus \{\bar t\}$ and $V(\bar
    t^-) \ge V(\bar t^+)$, where $L_K$ is the Lipschitz constant of
    $K$ and $\bar t \in [0, \sigma(0)]$ is the greatest time such that
    $w(t) = 0$ for all $t > \bar t$.}
\end{proposition}
\begin{proof}
  Using~\eqref{eq:Lf_S}, we have
  \begin{align}\label{eq:Lf_S2}
    \notag \!\!\!\!\!\!&\Lc_f S(x(t))
    \\
    \notag \!\!\!\!\!\!&\le \!-\gamma(|x(t)|) + \rho\big(|w(\phi(t))| \!+\!
    |K(x(t) \!+\! e(\phi(t))) \!-\! K(x(t))|\big)
    \\
    \notag \!\!\!\!\!\!&\le -\gamma(|x(t)|) + \rho\big(|w(\phi(t))| + L_K
    |e(\phi(t))|\big)
    \\
    \!\!\!\!\!\!&\le -\gamma(|x(t)|) + \rho(2|w(\phi(t))|) + \rho(2L_K
    |e(\phi(t))|).
  \end{align}
  \nnew{In the following, we provide a rigorous proof of the fact
    $\dot L(t) = -b L(t)$ stated in~\citep{NBL-MK:13}.  Similar to
    Lemma~8.9 therein, it holds that}
  \begin{align*}
    L(t) = \lim_{n \to \infty} \bigg[\int_t^{\sigma(t)} \!\!\!e^{2 n b
      (\tau - t)} w(\phi(\tau))^{2n} d \tau \bigg]^\frac{1}{2n}
    \!\!\triangleq \lim_{n \to \infty} L_n(t),
  \end{align*}
  since $e^{-b (t - \tau)} w(\phi(\tau))$ is bounded for $\tau \in [t,
  \sigma(t)]$ and any $t \ge 0$ and $[t, \sigma(t)]$ has finite
  measure.
  In fact, it can be shown that this convergence is uniform over $[0,
  t_1]$ for any $t_1 < \bar t$. Therefore, since $\dot L_n(t) = -b
  L_n(t) - \frac{L_n}{2n} \left(\frac{w(\phi(t))}{L_n}\right)^{2n}$,
  $\frac{w(\phi(t))}{L_n} < 1$ for $t \in [0, t_1]$ and sufficiently
  large $n$ and $b$, and $t_1 \in [0, \bar t)$ is arbitrary, it
  follows from~\citep[Thm 7.17]{WR:76} that $\dot L(t) = -b L(t)$ for
  $t \in (0, \infty) \setminus \{\bar t\}$.
 Combining this and~\eqref{eq:Lf_S2}, we get
  \begin{align*}
    \dot V(t) &\le -\gamma(|x(t)|) + \rho(2|w(\phi(t))|) + \rho(2L_K
    |e(\phi(t))|)
    \\
    &\quad + \frac{2}{b} 2 \dot L(t) \frac{\rho(2 L(t))}{2 L(t)}
    \\
    &\le -\gamma(|x(t)|) + \rho(2|w(\phi(t))|) + \rho(2L_K
    |e(\phi(t))|)
    \\
    &\quad - 2 \rho(2 L(t)).
  \end{align*}
  for $t \in (0, \infty) \setminus \{\bar
  t\}$. Equation~\eqref{eq:Vdot-bound} thus follows since
  $|w(\phi(t))| \le L(t)$ (c.f.~\eqref{eq:L}) and the fact that $\rho$
  is strictly increasing. Finally, since $S(x(t))$ is continuous,
  $L(\bar t^-) \ge 0$, and $L(\bar t^+) = 0$, we get $V(\bar t^-) \ge
  V(\bar t^+)$.
\end{proof}

Proposition~\ref{prop:ub-V-dot} is the basis for our event-trigger
design. Formally, we select $\theta \in (0, 1)$ and require 
\begin{align*}
  \rho(2L_K |e(\phi(t))|) \le \theta \gamma(|x(t)|), \qquad t \ge 0,
\end{align*}
which can be equivalently written as
\begin{align}\label{eq:trig} 
  |e(t)| \le \frac{\rho^{-1}(\theta \gamma(|p(t)|))}{2 L_K}, \qquad t
  \ge \phi(0).
\end{align}
Notice from~\eqref{eq:e} and the fact $t=0$ that~\eqref{eq:trig} holds
on $[\phi(0), t_0]$. Equation~\eqref{eq:trig} fully specifies the
sequence of times $\{t_k\}_{k = 1}^\infty$ and its dependence on the
actuation delay. For each~$k$, after each time $t_k$, the controller
keeps evaluating~\eqref{eq:trig} until it reaches equality. At this
time, labeled $t_{k+1}$, the controller triggers the next event that
sets $e(t_{k+1}) = 0$ and maintains~\eqref{eq:trig}. Notice that
``larger'' $\gamma$ and ``smaller'' $\rho$ (corresponding to
``stronger'' input-to-state stability in~\eqref{eq:iss-lyap2}) are
then more desirable, as they \new{are intuitively expected to let the
  controller} update $u$ less often. Our ensuing analysis shows global
asymptotic stability of the closed-loop system and the existence of a
uniform lower bound on the inter-event times.

\subsection{Convergence Analysis under Event-triggered
  Law}\label{subsec:trig}

In this section we show that our event triggered law~\eqref{eq:trig}
solves Problem~\ref{problem} by showing, in the following result, that
the inter-event times are uniformly lower bounded (so, in particular,
there is no finite accumulation point in time) and the closed-loop
system achieves global asymptotic stability.

\begin{theorem}\longthmtitle{Uniform Lower Bound for the
    Inter-Event Times and Global Asymptotic Stability}\label{thm:zeno&gas}
  Suppose that the class $\Kc_\infty$ function
    $\Gc: r \mapsto \gamma^{-1} (\rho(r)/\theta)$ is (locally) Lipschitz.
  For the system~\eqref{eq:dynamics} under the
  control~\eqref{eq:u2} and the triggering
  condition~\eqref{eq:trig}, the following hold:
  \begin{enumerate}
  \item there exists $\delta \new{= \delta(x(0), \{u(t)\}_{t =
        \phi(0)}^0)} > 0$ such that $t_{k+1} - t_k \ge \delta$ for all
    $k \ge 1$,
\item there exists $\beta \in
  \Kc\Lc$ such that for any $x(0) \in \real^n$ and bounded
  $\{u(t)\}_{t = \phi(0)}^0$, we have for all $t \ge 0$,
  \begin{align}\label{eq:cl-as}
    \hspace{-30pt} |x(t)| + \hspace{-3pt} \sup_{\phi(t) \le \tau \le
      t} \hspace{-3pt} |u(\tau)| \le \beta\Big(|x(0)| + \hspace{-3pt}
    \sup_{\phi(0) \le \tau \le 0} \hspace{-3pt} |u(\tau)|, t\Big).
  \end{align}
   \end{enumerate} 
\end{theorem}
\begin{proof}
  Let $[0, t_{\max})$ be the maximal interval of existence of the
  solutions of the closed-loop system. The proof involves three
  steps. First, we prove that (ii) holds for $t < t_{\max}$. Then, we
  show that (i) holds until $t_{\max}$, and finally that $t_{\max} =
  \infty$.

  \emph{Step 1:} From Proposition~\ref{prop:ub-V-dot}
  and~\eqref{eq:trig}, we have
  \begin{align*}
    \dot V(t) &\le -(1 - \theta) \gamma(|x(t)|) - \rho(2 L(t)) \\
    &\le -\gamma_\text{min}(|x(t)| + L(t)), \qquad t \in [0, t_{\max})
    \setminus \{\bar t\},
  \end{align*}
  where $\gamma_\text{min}(r) = \min\{(1 - \theta) \gamma(r), \rho(2
  r)$ for all $r \ge 0$, so $\gamma_\text{min} \in \Kc$. Also, note
  that
  \begin{align*}
    V(t) \le \alpha_2(|x(t)|) + \alpha_0(L(t)) \le 2
    \alpha_\text{max}(|x(t)| + L(t)),
  \end{align*}
  where $\alpha_\text{max}(r) = \max\{\alpha_2(r), \alpha_0(r)\}$ and
  $\alpha_0(r) = \frac{2}{b} \int_0^{2r} \frac{\rho(s)}{s} d s$ for
  all $r \ge 0$. Since $\alpha_0, \alpha_2 \in \Kc_\infty$, we have
  $\alpha_\text{max} \in \Kc_\infty$, so $\alpha_\text{max}^{-1} \in
  \Kc$. Hence,
  \begin{align*}
    \dot V(t) \le -\alpha_\text{min}(\alpha_\text{max}^{-1}(V(t)/2))
    \triangleq \overline \alpha(V(t)), \ t \in [0, t_{\max}) \setminus \{\bar t\},
  \end{align*}
  where $\overline \alpha \in \Kc$. Therefore, using the Comparison
  Principle~\citep[Lemma 3.4]{HKK:02}, \citep[Lemma 4.4]{HKK:02}, and
  $V(\bar t^-) \ge V(\bar t^+)$, there exists $\beta_1 \in \Kc\Lc$
  such that $ V(t) \le \beta_1(V(0), t)$, $ t < t_{\max}$.  Therefore,
  \begin{align*}
    |x(t)| + L(t) \le \beta_2(|x(0)| + L(0), t), \qquad t < t_{\max},
  \end{align*}
  where $\beta_2(r, s) = \alpha_\text{min}^{-1}(\overline \beta(2
  \alpha_\text{max}(r), s))$ for any $r, s \ge 0$. Note that $\beta_2
  \in \Kc\Lc$. Since we have
  \begin{align*} 
    \sup_{\phi(t) \le \tau \le t} |w(\tau)| \le L(t) \le e^{b M_0}
    \sup_{\phi(t) \le \tau \le t} |w(\tau)|,
  \end{align*}
  it then follows that
  \begin{align}\label{eq:x-bound}
    \!\!\!\! |x(t)| + \!\!\! \sup_{\phi(t) \le \tau \le t} \!\!\!\! |w(\tau)| \le
    \beta_3\Big(|x(0)| + \!\!\! \sup_{\phi(0) \le \tau \le 0} \!\!\!\! |w(\tau)|,
    t\Big),
  \end{align} 
  for all $t < t_{\max}$, where $\beta_3(r, s) = \beta_2(e^{b M_0} r,
  s)$. This inequality leads to~\eqref{eq:cl-as} using the same
    steps as in~\citep[Lemmas 8.10, 8.11]{NBL-MK:13} (the only
    difference being the multiplicity of inputs).
  
  \emph{Step 2:} Equation~\eqref{eq:trig} can be rewritten as
  \begin{align*}
    |p(t)| \ge \gamma^{-1} \Big(\frac{\rho(2 L_K
        |e(t)|)}{\theta}\Big).
  \end{align*}
  From step 1,
  the prediction $p(t) = x(\sigma(t))$ and its error
  $e(t) = p(t_k) - p(t)$ are bounded. Therefore, there exists
  $L_{\gamma^{-1} \rho/\theta}>0$ such that for all $t \ge 0$,
  \begin{align*}
    \gamma^{-1} \Big(\frac{\rho(2 L_K |e(t)|)}{\theta}\Big) \le 2
  L_{\gamma^{-1} \rho/\theta} L_K|e(t)|.
  \end{align*}
  where $L_{\gamma^{-1} \rho/\theta}$ is the Lipschitz constant of
  $\Gc$ on the compact set that contains $\{e(t)\}_{t =
    0}^{t_{\max}}$. Hence, a sufficient (stronger) condition
  for~\eqref{eq:trig} is
  \begin{align}\label{eq:trig2}
    |p(t)| \ge 2 L_{\gamma^{-1} \rho/\theta} L_K |e(t)|.
  \end{align}
  Note that~\eqref{eq:trig2} is only for the purpose of analysis and
  is \emph{not} executed in place of~\eqref{eq:trig}.  Clearly, if the
  inter-event times of~\eqref{eq:trig2} are lower bounded, so are the
  inter-event times of~\eqref{eq:trig}.  Let $r(t) =
  \frac{|e(t)|}{|p(t)|}$ for any $t \ge 0$ (with $r(t) = 0$ if $p(t) =
  0$). For any $k \ge 0$, we have $r(t_k) = 0$ and $t_{k+1} - t_k$ is
  greater than or equal to the time that it takes for $r(t)$ to go
  from $0$ to $\frac{1}{2 L_{\gamma^{-1} \rho/\theta} L_K}$. Note that
  for any $t \ge 0$,
  \begin{align*}
    \dot r &= \frac{d}{d t} \frac{|e|}{|p|} = \frac{d}{d t} \frac{(e^T
      e)^{1/2}}{(p^T p)^{1/2}}
    \\
    &= \frac{(e^T e)^{-1/2} e^T \dot e (p^T p)^{1/2} - (p^T p)^{-1/2}
      p^T \dot p (e^T e)^{1/2}}{p^T p}
    \\
    &= -\frac{e^T \dot p}{|e| |p|} - \frac{|e| p^T \dot p}{|p|^3} \le
    \frac{|\dot p|}{|p|} + \frac{|e| |\dot p|}{|p|^2} = (1 + r)
    \frac{|\dot p|}{|p|},
  \end{align*}
  where the time arguments are dropped for better readability. To
  upper bound the ratio $|\dot p(t)|/|p(t)|$, we have
  from~\eqref{eq:p} that $\dot p(t) = \dot \sigma(t) f(p(t), u(t))$
  for all $t \ge \phi(0)$. By continuous differentiability of $f$
  (which implies Lipschitz continuity on compacts) and global
  asymptotic stability of the closed loop system, there exists
  $L_f > 0$ such that
  \begin{align*}
    |\dot p(t)| &= |\dot \sigma(t) f(p(t), u(t))| \le M_2 |f(p(t),
    K(p(t) + e(t)))|
    \\
    &\le M_2 L_f |(p(t), K(p(t) + e(t)))|
    \\
    &\le M_2 L_f(|p(t)| + |K(p(t) + e(t))|)
    \\
    &\le M_2 L_f(|p(t)| + L_K|p(t) + e(t)|)
    \\
    &\le M_2 L_f(1 + L_K) |p(t)| + M_2 L_f L_K |e(t)|
    \\
    \Rightarrow \dot r(t) &\le M_2 (1 + r(t))(L_f(1 + L_K) + L_f L_K
    |r(t)|).
  \end{align*}
  Thus, using the Comparison Principle~\citep[Lemma
  3.4]{HKK:02}, we have $t_{k+1} - t_k \ge \delta, k \ge 0$
  where $\delta$ is the time that it takes for the solution of
  \begin{align}\label{eq:r-dot}
    \dot r = M_2 (1 + r)(L_f(1 + L_K) + L_f L_K r),
  \end{align}
  to go from $0$ to $\frac{1}{2 L_{\gamma^{-1} \rho/\theta} L_K}$.
 
  \emph{Step 3:} Since all system trajectories are bounded and $t_k
  \xrightarrow{k \to \infty} \infty$, we have $t_{\max} = \infty$,
  completing the proof.
\end{proof}

A particular corollary of Theorem~\ref{thm:zeno&gas} is that the
proposed event-triggered law does not suffer from Zeno behavior, i.e.,
$t_k$ accumulating to a finite point $t_{\max}$. Also, note that the
lower bound $\delta$ in general depends on the initial conditions
$x(0)$ and $\{u(t)\}_{t = \phi(0)}^0$ through the Lipschitz constant
$L_{\gamma^{-1} \rho/\theta}$.%
\footnote{However, for any given compact set of $|x(0)|$ and $|u(t)|, t < 0$, equations~\eqref{eq:x-bound}, \eqref{eq:w}, \eqref{eq:e}, and~\eqref{eq:p} ensure that $x(t)$ and therefore $p(t)$ are bounded for all $t$, and so $e(t)$ belongs to a compact set due to~\eqref{eq:trig}. Hence, $L_{\gamma^{-1} \rho/\theta}$ and thus $\delta$ can be chosen uniformly over this set.}
\nnew{Finally, while Theorem~\ref{thm:zeno&gas} explicitly bounds $x$ and $u$, the simple time-shift relationship~\eqref{eq:p} between $p$ and $x$ ensures that any bound satisfied by $x(t), t \ge 0$, including that of Theorem~\ref{thm:zeno&gas}, is also satisfied by $p(t), t \ge 0$.}

\subsection{Delayed and Event-Triggered Sensing}\label{sec:sensing}

So far, we have not considered any delays in the availability of the
sensing information about the plant state, \new{which we consider
  next. Our treatment here shows that the above event-triggered
  controller with \emph{the same triggering condition~\eqref{eq:trig},
    and with slight adjustments in the employed control and predictor
    signals,} globally asymptotically stabilizes the plant while
  maintaining the same lower bound on the inter-event times.}

To address the general scenario in Problem~\ref{problem},~let
\begin{align*}
  \bar \ell = \bar \ell(t) = \max\setdef{\ell \ge 0}{\tau_\ell \le
    \psi(t)},
\end{align*} 
be the index of the last plant state available at the controller
at time~$t$. Then, \eqref{eq:p} is replaced with
\begin{align}\label{eq:p-psi}
 \!\!\!\!p(t) = x(\tau_{\bar \ell}) + \!\! \int_{\phi(\tau_{\bar \ell})}^t
  \!\!\!\!\! \dot \sigma(s) f(p(s), u(s)) d s, \ t \ge
  \psi^{-1}(0),
\end{align}
which is the best estimate of $x(\sigma(t))$ available to the controller%
\footnote{This only requires the controller to know
  $\psi(\tau_\ell)$ for every received state (not the full function
  $\psi$), which is realized by having a time-stamp
  for~$x(\tau_\ell)$.}.  Since $p(t)$ is not available before
$\psi^{-1}(0)$, the control signal~\eqref{eq:u1},~\eqref{eq:u2} is
updated as
\begin{align}\label{eq:u-new}
  u(t) = \begin{cases}
    K(p(t_k))  \quad &\text{if} \quad t \in [t_k, t_{k+1}), \ k \ge 0, \\
    0 \quad &\text{if} \quad t \in [0, t_0),
  \end{cases}
\end{align}
where the first event time is now $t_0 = \psi^{-1}(0)$.
We next provide the same guarantees as Theorem~\ref{thm:zeno&gas}.

\begin{theorem}\label{thm:psi}
  Consider the plant dynamics~\eqref{eq:dynamics} driven by the
  predictor-based event-triggered controller~\eqref{eq:u-new} with
  the predictor~\eqref{eq:p-psi} and triggering
  condition~\eqref{eq:trig}. Under \new{Assumption~\ref{assum}}, 
  the closed-loop system is globally asymptotically stable,
  namely, there exists $\beta \in \Kc\Lc$ such that~\eqref{eq:cl-as}
  holds for all $x(0) \in \real^n$, continuously differentiable $\{u(t)\}_{t =
    \phi(0)}^0$, and $t \ge 0$. Furthermore, there exists $\delta \new{= \delta(x(0), \{u(t)\}_{t = \phi(0)}^0)} > 0$ such that 
  $t_{k+1} - t_k \ge \delta$ for all $k \ge 0$.
\end{theorem}
\begin{proof}
   For simplicity, let $U(t) = \sup_{\phi(t) \le \tau \le t} |u(t)|$.
  Since the open-loop system exhibits no finite escape time behavior,
  the state remains bounded during the initial period $[0,
  t_0]$. Hence, for any $x(0)$ and any $\{u(t)\}_{t = \phi(0)}^0$
  there exists $\Xi > 0$ such that $|x(t)| \le \Xi$ for $t \in [0,
  t_0]$. Without loss of generality, $\Xi$ can be chosen to be a class
  $\Kc$ function of $|x(0)| + U(0)$. Thus,
  \begin{align}\label{eq:xuK}
    |x(t)| + &U(t) \le \Xi(|x(0)| + U(0)) + U(0)
    \\
    \notag &\ \ \le \big[\Xi(|x(0)| + U(0)) + U(0)\big] e^{-(t - t_0)}, \quad t \in [0, t_0].
  \end{align}
  As soon as the controller receives $x(0)$ at $t_0$, it can estimate
  the state $x(t)$ by simulating the dynamics~\eqref{eq:dynamics},
  i.e.,
  \begin{align}\label{eq:est}
    x(t) = x(0) + \int_{0}^t f(x(s), u(\phi(s))) d s.
  \end{align}
  This estimation is updated whenever a new state $x(\tau_\ell)$
  arrives and used
    to compute the predictor~\eqref{eq:p}, which combined
    with~\eqref{eq:est} takes the form~\eqref{eq:p-psi}. Since the
    controller now has access to the same prediction signal $p(t)$ as
    before, the same Lyapunov analysis as above holds for $[t_0,
    \infty)$. Therefore, let $\hat \beta \in \Kc \Lc$ be such
    that~\eqref{eq:cl-as} holds for $t \ge t_0$. By~\eqref{eq:xuK},
    \begin{align*}
      |x(t)| + U(t) \le \hat \beta \big(\Xi(|x(0)| + U(0)) + U(0), t - t_0\big) \quad t \ge t_0.
    \end{align*}
    Therefore, \eqref{eq:cl-as} holds by choosing
    $\beta(r, t) = \max\big\{\hat \beta \big(\Xi(r) + r, t - t_0\big), \big[\Xi(r) + r\big] e^{-(t - t_0)}\big\}$.
    Finally, since the triggering condition~\eqref{eq:trig} has not
    changed, $t_{k + 1} - t_k \ge \delta, k \ge 0$ for the same
    $\delta > 0$ as in Theorem~\ref{thm:zeno&gas}.
  \end{proof}

  \new{
    \begin{remark}\longthmtitle{Separation of sensing and actuation
        delays}
      \rm It is a standard practice in the literature to combine the
      sensing and actuation delays into a single quantity, i.e.,
      ``networked induced delays". This is in fact the basis of the
      predictor design in equation~(23). However, in our treatment, it
      is beneficial to keep the two delays distinct since their
      sources are often physically distinct and the assumptions on the
      sensing delay $\psi$ are significantly weaker than on the
      actuator delay $\phi$ (cf. Assumption~\ref{assum}). \oprocend
    \end{remark}
  }

  \new{
    \begin{remark}\longthmtitle{Practical importance of
        feedback}\label{rem:fb}
  \rm While the controller can theoretically
  discard 
  $\{x(\tau_\ell)\}_{\ell = 1}^\infty$ and rely on $x(0)$ for
  estimating the state at all future times, closing the loop using the
  most recent state value $x(\tau_{\bar \ell})$ is in practice
  critical for 
  preventing the estimator~\eqref{eq:est} from drifting due to noise
  and un-modeled dynamics, even when the system dynamics are perfectly
  known. This is apparent, for instance, in Example~\ref{ex:1} shown
  later, where facing the errors caused by the numerical approximation
  of the prediction signal. \oprocend
\end{remark}
}

\section{The Linear Case}

\nnew{Here, we specialize the general treatment of
  Section~\ref{sec:control} to the linear case
\begin{align}\label{eq:lin-dynamics}
\!\!\!\dot x(t) = A x(t) + B u(\phi(t)), \quad \text{a.a. } t \ge 0, \quad x(0) = x_0.
\end{align}
For simplicity, we restrict our attention
to the perfect sensing case, with similar generalizations to sampled and delayed sensing as in Section~\ref{sec:sensing}.}
Assuming that the pair $(A, B)$ is stabilizable, we can use pole
placement to find a linear feedback law $K:\real^n \to \real$ that
\new{satisfies
  Assumption~\ref{assum}\ref{item:assum-iss}.} 
Moreover, $p(t)$ \nnew{has the explicit
  form} 
\begin{align}\label{eq:lin-p}
  p(t) = \new{e^{A(\sigma(t) - t_+)} x(t_+) + \int_{\phi(t_+)}^t} \!\!\dot
  \sigma(s) e^{A(\sigma(t) - \sigma(s))} B u(s) d s,
\end{align}
for all $t \ge \phi(0)$ and the closed-loop system takes the form
\begin{align*}
  \dot x(t) &= (A + BK) x(t) + B w(\phi(t)) + B K e(\phi(t)).
\end{align*}
Furthermore, given an arbitrary $Q = Q^T > 0$, the continuously
differentiable function $S: \real^n \to \real$ is $ S(x) = x^T P x$,
where $P = P^T > 0$ is the unique solution to the Lyapunov equation $
(A + B K)^T P + P (A + B K) = -Q$.  Clearly,~\eqref{eq:alpha12} holds
with $\alpha_1(r) = \lambda_\text{min}(P) r^2$ and $\alpha_2(r) =
\lambda_\text{max}(P) r^2$. Also, using
Young's inequality~\citep{WHY:12},
\begin{align*}
 \Lc_f S = -x(t)^T Q x(t)
  + 2 x(t)^T P B (w(\phi(t)) + K e(\phi(t))),
\end{align*}
so~\eqref{eq:Lf_S} holds with $\gamma(r) = \frac{1}{2}
\lambda_\text{min}(Q) r^2$ and $\rho(r) = \frac{2 |P
  B|^2}{\lambda_\text{min}(Q)} r^2$. \nnew{Thus,
\eqref{eq:trig} also takes the simpler form
\begin{align}\label{eq:lin-trig}
  |e(t)| \le \frac{\lambda_\text{min}(Q) \sqrt\theta}{4 |P B| |K|}
  |p(t)|.
\end{align}
In addition to these simplifications, we show next that the closed-loop
system is globally exponentially stable.} 
  
\subsection{Exponential Stability under Event-triggered Control}

\nnew{We next show that, in the linear case, we obtain the stronger feature
of global exponential stability using a slightly
different Lyapunov-Krasovskii functional.}

\begin{theorem}\longthmtitle{Exponential Stabilization}\label{thm:es}
  The system~\eqref{eq:lin-dynamics} subject to the piecewise-constant
  closed-loop control $ u(t) = K p(t_k)$, $ t \in [t_k, t_{k+1})$,
  with $p(t)$ given in~\eqref{eq:lin-p} and $\{t_k\}_{k=1}^\infty$
  determined according to~\eqref{eq:lin-trig} satisfies
  \begin{align*}
    |x(t)|^2 + \int_{\phi(t)}^t \!\!\!u(\tau)^2 d \tau \le C e^{-\mu
      t} \Big(|x(0)|^2 + \int_{\phi(0)}^0 \!\!\!u(\tau)^2 d \tau
    \Big)\!,
  \end{align*}
  for some $C > 0$, $\mu = \frac{(2 - \theta) \lambda_\text{min}(Q)}{4
    \lambda_\text{max}(P)}$, and all $t \ge 0$.
\end{theorem}
\begin{proof}
  For $t \ge 0$, let
    $L(t) = \int_t^{\sigma(t)} e^{b(\tau - t)} w(\phi(\tau))^2 d \tau$.
  One can see that $ \dot L(t) = -w(\phi(t))^2 - b L(t)$, $ t \ge 0$.
  Define
    $V(t) = x(t)^T P x(t) + \frac{4 |P B|^2}{\lambda_\text{min}(Q)}
    L(t)$.
  Therefore, using~\eqref{eq:lin-trig},
  \begin{align*}
    \dot V(t) &= -x(t)^T Q x(t) + 2 x(t)^T P B w(\phi(t)) - \frac{4 |P B|^2 b}{\lambda_\text{min}(Q)} L(t)
    \\
    &\quad + 2 x(t)^T P B K e(\phi(t)) - \frac{4 |P
      B|^2}{\lambda_\text{min}(Q)} w(\phi(t))^2
\\
    &\le -\frac{2 - \theta}{4} \lambda_\text{min}(Q) |x(t)|^2 -
    \frac{4 |P B|^2 b}{\lambda_\text{min}(Q)} L(t)
    \le - \mu V(t),
  \end{align*}
  where $\mu = \min\big\{\frac{(2 - \theta) \lambda_\text{min}(Q)}{4
    \lambda_\text{max}(P)}, b\big\} = \frac{(2 - \theta)
    \lambda_\text{min}(Q)}{4 \lambda_\text{max}(P)}$ if $b$ is chosen
  sufficiently large. Hence, by the Comparison Principle~\citep[Lemma
  3.4]{HKK:02}, we have $ V(t) \le e^{-\mu t} V(0)$, $ t \ge 0$.  Let
  $W(t) = |x(t)|^2 + \int_{\phi(t)}^t u(\tau)^2 d
  \tau$. From~\citep[Eq. (6-99)-(6-100)]{NBL-MK:13}, $c_1 W(t) \le
  V(t) \le c_2 W(t)$, for some $c_1, c_2 > 0$ and all $t \ge
  0$. Hence, the result follows with $C = {c_2}/{c_1}$.
\end{proof}

From Theorem~\ref{thm:es}, the convergence rate $\mu$ depends both on
the ratio $\frac{\lambda_\text{min}(Q)}{\lambda_\text{max}(P)}$ and
the parameter $\theta$. The former can be increased by placing the
eigenvalues of $A + B K$ at larger negative values, though large
eigenvalues result in noise amplification. Decreasing $\theta$,
however, comes at the cost of faster control updates, a trade-off we
study~next.

\subsection{Optimizing the Sampling-Convergence Trade-off}

\nnew{Here, we analyze the trade-off between sampling frequency and
  convergence speed. 
  In general, it is clear from the Lyapunov analysis of
  Section~\ref{sec:control} that more updates \nnew{(intuitively
    corresponding to smaller $\theta$)} hasten the decay of $V(t)$ and
  help convergence.
Let $\delta$ be the time that it takes
for the solution of~\eqref{eq:r-dot} to go from $0$ to $\frac{1}{2
  L_{\gamma^{-1} \rho/\theta} L_K}$. As shown in
Section~\ref{subsec:trig}, the inter-event times are lower bounded by
$\delta$, so it can be used to bound the sampling cost of
  implementing the controller.}  Let
\begin{align*}
  a = M_2 L_f L_K, \quad\!\! c = M_2 L_f (1 + L_K), \quad\!\! R =
  \frac{1}{2 L_{\gamma^{-1} \rho/\theta} L_K},
\end{align*}
where $L_f = \sqrt2 (|A| + |B|)$, $L_K = |K|$, and $L_{\gamma^{-1}
  \rho/\theta} = \frac{2 |P B|}{\lambda_\text{min}(Q)
  \sqrt\theta}$. Then, the solution of~\eqref{eq:r-dot} with initial
condition $r(0) = 0$ is 
$r(t) = \frac{c e^{a t} - c e^{c
    t}}{a e^{c t} - c e^{a t}}$, so solving $r(\delta) = R$ for $\delta$
gives
  $\delta = \frac{\ln\frac{c + R a}{c + R c}}{a - c}$.
The objective is to maximize $\delta$ and $\mu$ by tuning the
optimization variables $\theta$ and $Q$. For simplicity, let
 $ \theta = \nu^2$ and $Q = q I_n$
 where $\nu, q > 0$.  Then,
\begin{align*}
  \delta(\nu) = \frac{1}{a - c} \ln \frac{c + \frac{\nu}{|P_1 B| |K|}
    a}{c + \frac{\nu}{|P_1 B| |K|} c}, \quad \mu(\nu) = \frac{2 -
    \nu^2}{4 \lambda_\text{max}(P_1)},
\end{align*}
where $P_1 = q^{-1} P$ is the solution of the Lyapunov equation $(A + B
K)^T P_1 + P_1 (A + B K) = -I_n$. Figure~\ref{fig:opt}(a) depicts
$\delta$ and $\mu$ as functions of $\nu$ and illustrates the
sampling-convergence trade-off.

\begin{figure}
  \centering
  \includegraphics[width = 0.505\linewidth]{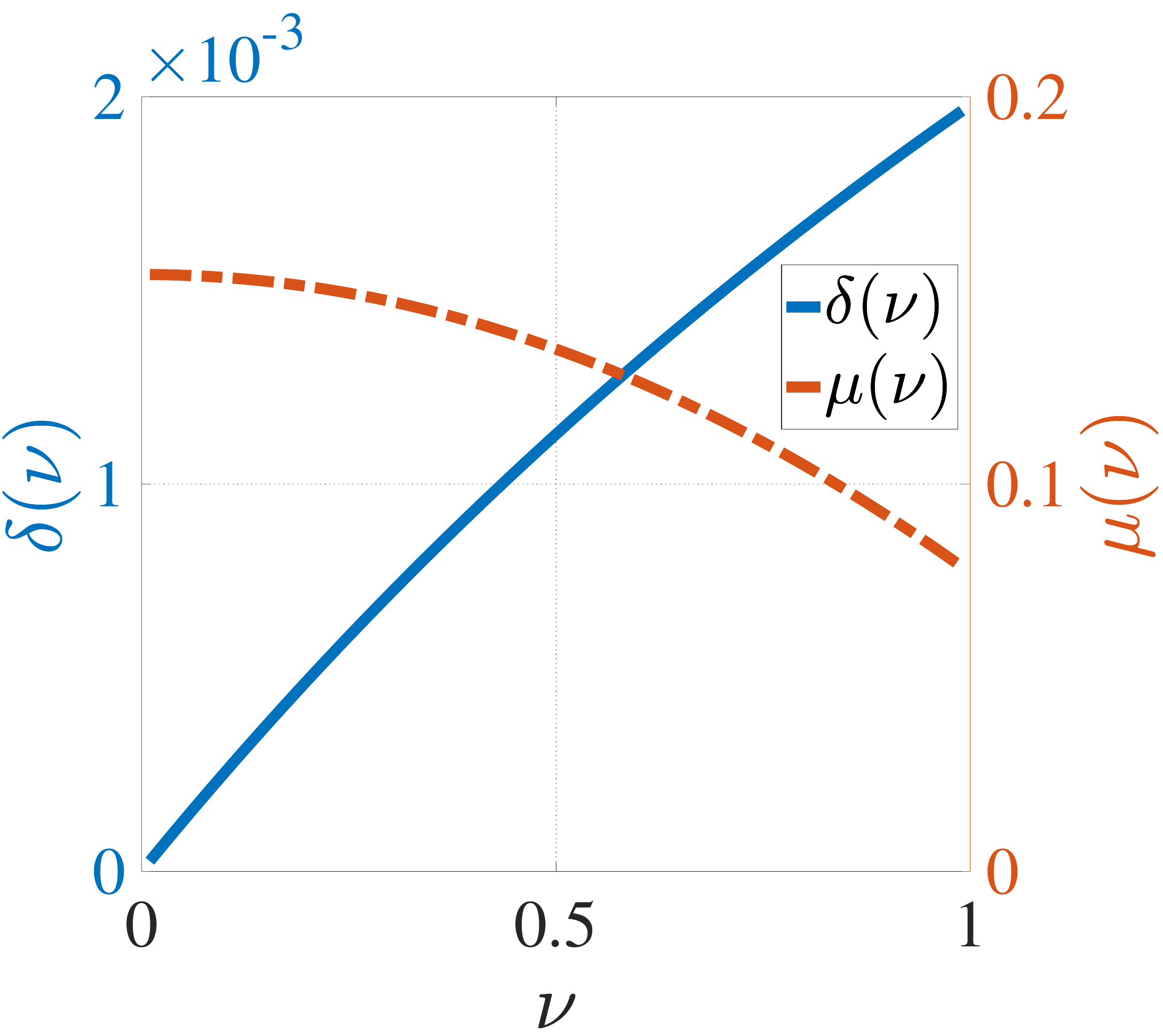}
  \includegraphics[width = 0.425\linewidth]{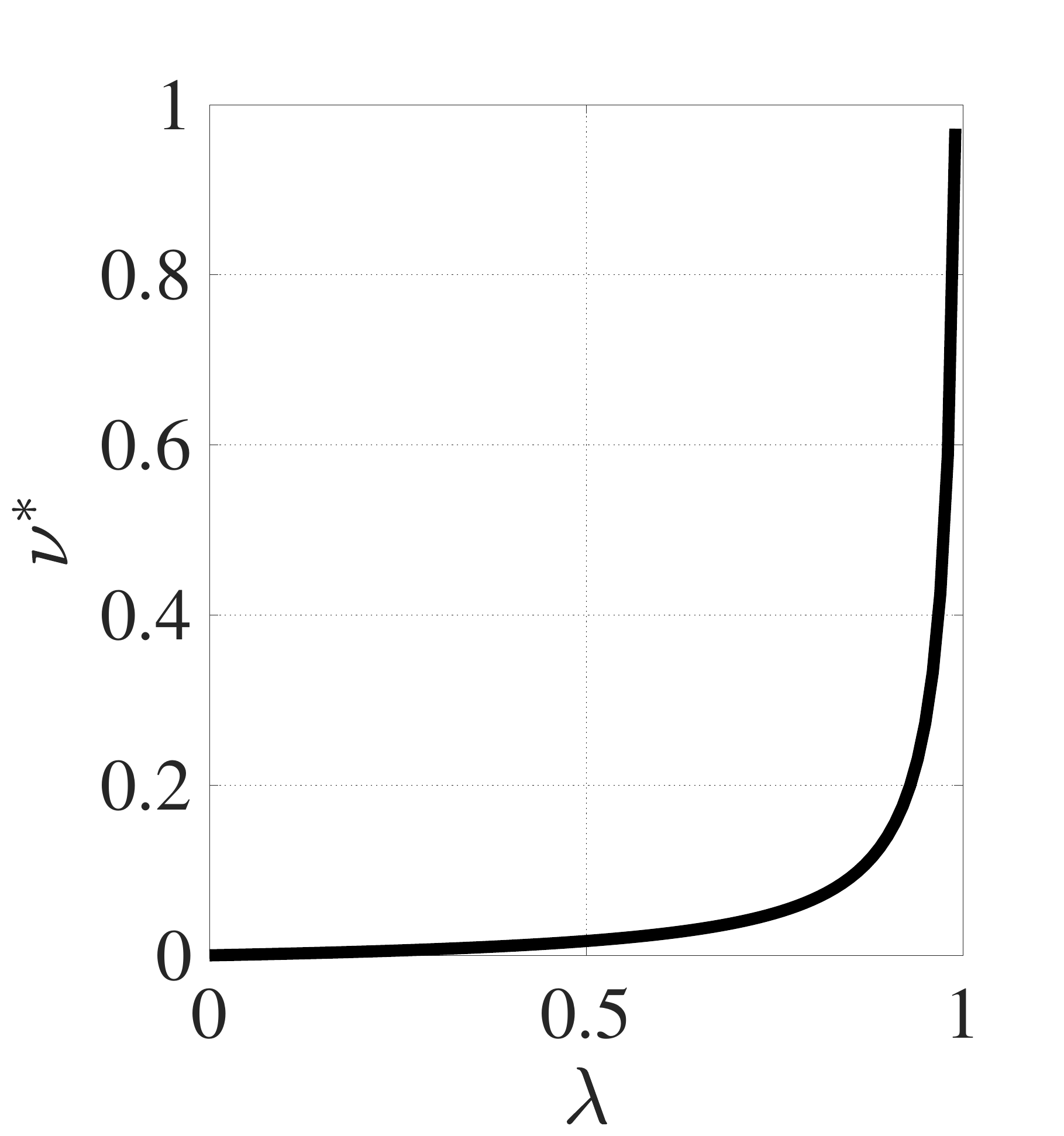}
  \caption{Sampling-convergence trade-off for event-triggered control
    of linear systems. %
    \textbf{(Left)}, values of the lower bound of the inter-event
    times ($\delta$) and exponential rate of convergence ($\mu$) for
    different values of the optimization parameter $\nu$ for a
    3rd-order unstable linear system with $M_2 = 1$. %
    \textbf{(Right)}, unique maximizer $\nu^*$ of the objective
    function $J(\nu)$ for different values of~$\lambda$. As $\lambda$
    goes from $0$ to $1$, more weight is given to the maximization of
    $\delta$, which increases $\nu^*$.}\label{fig:opt}
\vspace*{-1ex}
\end{figure}

To balance these two objectives, we define the aggregate objective
function as a convex combination of $\delta$ and $\mu$, 
\begin{align*}
  J(\nu) = \lambda \delta(\nu) + (1 - \lambda) \mu(\nu),
\end{align*}
\nnew{where $\lambda \in [0, 1]$ determines the relative importance of
  convergence rate and sampling.
  The function $J$ is strongly convex and its unique maximizer is the
  positive real solution of
  $c_3 \nu^3 + c_2 \nu^2 + c_1 \nu + c_0 = 0$,
  where $c_3 = a (1 - \lambda)$, $c_2 = (a + c)|P_1 B| |K|(1 -
  \lambda)$, $c_1 = c |P_1 B|^2 |K|^2(1 - \lambda)$, and $c_0 = -2
  \lambda_\text{max}(P_1) |P_1 B| |K|
  \lambda$. Figure~\ref{fig:opt}(b) plots this
  maximizer 
  for different values of 
  $\lambda$, further illustrating the sampling-convergence trade-off.}

\section{Simulations}\label{sec:sims}

Here we illustrate the performance of our event-triggered
predictor-based design. Example~\ref{ex:1} is a two-dimensional
nonlinear system that satisfies all the hypotheses required to ensure
global asymptotic convergence of the closed-loop
system. Example~\ref{ex:2} is a different two-dimensional nonlinear
system \nnew{which, instead,} does not but for which we observe convergence in
simulation. We start by discussing some numerical challenges
  that arise because of the particular hybrid nature of our design,
  along with our approach to tackle them.

 \begin{remark}\longthmtitle{Numerical implementation of event-triggered
     control law}\label{rem:num}
   \rm The main challenge in the numerical simulation of the proposed
     event-trigger law is the computation of the prediction signal
     $p(t) = x(\sigma(t))$. To this end, at least three methods can be
     used, as follows:\newline%
     (i) \emph{Open-loop}: One can solve $\dot p(t) = \dot \sigma(t)
     f(p(t), u(t))$ directly starting from $p(\phi(0)) = x(0)$. The
     closed-loop system takes the form of a hybrid system
     \new{(see, e.g.,~\citep{RG-RGS-ART:12} for an introduction to
       hybrid systems)} with flow map
   \begin{subequations}\label{eq:perturb}
     \begin{align}
       \label{eq:x-dot} \dot x(t) &= f(x(t),
       u(\phi(t))), 
       \\
       \label{eq:p-dot} \dot p(t) &= \dot \sigma(t) f(p(t), u(t)), 
       \\
       \dot p_{tk}(t) &= 0, 
       \\
       u(t) &= K(p_{tk}(t)), 
     \end{align}
   \end{subequations}
   jump map \new{$p_{tk}((t_k)_+) = p((t_k)_+)$}, jump set $D =
   \setdefb{(x, p, p_{tk})}{|p_{tk} - p| = \frac{\rho^{-1}(\theta
       \gamma(|p|))}{2 L_K}}$,
    and flow set $C = \overline{\real^{3n} \setminus D}$.  \nnew{Note
     that $x$ and $p$ do not change at jumps (i.e., identity
     maps). Here, the values of $p$ (resp. $p_{tk}$ and $u$) of any
     hybrid solution are arbitrary in the interval $[0,\phi(0))$
     (resp. $[0,t_0)$).}
   This formulation is computationally efficient but, if the original
   system is unstable, it is prone to numerical instabilities. The
   reason, suggesting the name ``open-loop'', is that the $(p,
   p_{tk})$-subsystem is completely decoupled from the
   $x$-subsystem. Therefore, as stated in Remark~\ref{rem:fb}, if any
   mismatch occurs between $x(t)$ and $p(\phi(t))$ due to numerical
   errors, the $x$-subsystem tends to become unstable, and this is not
   ``seen'' by the $(p, p_{tk})$-subsystem.
   \newline
   (ii) \emph{Semi-closed-loop}: One can add a feedback path from the
   $x$-subsystem to the $(p, p_{tk})$ subsystem by computing $p$
   directly from~\eqref{eq:p-psi} \new{every time a new state value
     arrives (i.e., at every
     $\psi^{-1}(\tau_\ell)$).} 
   This requires a numerical integration of $f(p(s), u(s))$ over the
   ``history'' of $(p, u)$ from $\phi(\tau_{\bar \ell})$ to $t$. This
   method is more computationally expensive but improves the numerical
   robustness. However, since we are still integrating over the
   history of $p$, any mismatch in the prediction takes more time to
   die out, which may not be tolerable for an unstable system.
   \newline
   (iii) \emph{Closed-loop}: To further increase robustness, one can
   solve the differential form in~\eqref{eq:p-dot} rather than the
   integral form in~\eqref{eq:p-psi} \new{every time a new state value
     arrives (i.e., at every
     $\psi^{-1}(\tau_\ell)$)} 
   from $\phi(\tau_{\bar \ell})$ to $t$ with ``initial'' condition
   $p(\phi(\tau_{\bar \ell})) = x(\tau_{\bar \ell})$. \new{This method
     is as computationally expensive as (ii) but is considerably more
     robust. This is therefore the recommended method for the
     numerical implementation of the proposed predictor-based
     controller and used below} in Examples~\ref{ex:1}
   and~\ref{ex:2}.  \oprocend
 \end{remark}

\begin{example}\longthmtitle{Compliant Nonlinear System}\label{ex:1}
  {\rm 
  Consider the 2-dimensional system given by
\begin{align*}
    f(x, u) &= \begin{bmatrix}
        x_1 + x_2
        \\
        \tanh(x_1) + x_2 + u
      \end{bmatrix}, \ \ \phi(t) = t - \frac{(t - 5)^2 + 2}{2 (t - 5)^2 + 2},
    \\
    \tau_\ell &= \ell \Delta_\tau, \quad \ell \ge 0, \hspace{42pt} \psi(t) = t - D_\psi,
  \end{align*}
  where $\Delta_\tau$ and $D_\psi$ are constants. This system
  satisfies
  \new{Assumption~\ref{assum}} 
  with the feedback law $K(x) = -6x_1 - 5 x_2 - \tanh(x_1)$, $S(x) =
  x^T P x$, and
  \begin{align*}
    & \new{L_f = \frac{\sqrt{2\sqrt{17}+10}}{2}, \quad\! L_K = \sqrt{74}} , \quad\!\! (M_1, m_2) = 1 \pm \frac{3 \sqrt 3}{16},
    \\
    & M_0 = 1, \quad \gamma(r) =
    \frac{\lambda_\text{min}(Q)}{2} r^2, \quad \rho(r) = \frac{2 |P
      B|^2}{\lambda_\text{min}(Q)} r^2,
  \end{align*}
  where $P = P^T > 0$ is the solution of $(A + B k)^T P + P (A + B k)
  = -Q$ for $A = [1 \ 1;\ 0 \ 1]$, $B = [0;\ 1]$, $k = [-6 \ -5]$, and
  arbitrary $Q = Q^T > 0$ \new{(we use $Q = I$)}. A sample simulation
  result of this system is depicted in Figure~\ref{fig:ex1}(a).
  It is to be noted that for this example, \eqref{eq:trig} simplifies
  to $|e(t)| \le \overline \rho |p(t)|$ with \new{$\overline \rho = 0.022$},
  but the closed-loop system remains stable when increasing \new{$\overline
  \rho$ about until $0.8$ (Figure~\ref{fig:ex1}(b))}. 
  
\new{While Theorem~\ref{thm:zeno&gas} guarantees the global asymptotic
  stability of the continuous-time system, discretization
  accuracy/error plays an important role in its digital
  implementation. It is with this in mind that one should interpret
  Figure~\ref{fig:ex1}(c), where depending on the discretization
  scheme and the stepsize employed, the numerical approximation errors
  in computing the prediction signal, cf. Remark~\ref{rem:num}, make
  the evolution of the Lyapunov function $V$ not monotonically
  decreasing (whereas we know from Theorem~\ref{thm:zeno&gas} that it
  is monotonically decreasing for the continuous-time system).  We see
  that, at least for this example, the effect on the evolution of $V$
  is sensitive to both the order of discretization and the stepsize
  ($h$), and benefits more from decreasing the latter.}
  
\new{Stability is also critically dependent on the sensing sampling
  rate $1/\Delta_\tau$, as noted in Remark~\ref{rem:fb}. We can also
  see from Figure~\ref{fig:ex1}(c) that the decay of $V$ clearly
  deteriorates for large $\Delta_\tau$ (insufficient sampling) due to
  (in this example only discretization) noise but can be made
  monotonic for sufficiently small $\Delta_\tau$.  To visualize this
  effect on stability more systematically}, we varied $\Delta_\tau$
and $D_\psi$ and computed $|x(25)|$ as a measure of asymptotic
stability. The average result is depicted in Figure~\ref{fig:ex1}(d)
for $10$ random initial conditions, showing that unlike our
theoretical expectation, large $\Delta_\tau$ and/or $D_\psi$ result in
instability even in the absence of noise because of the numerical
error that degrades the estimation~\eqref{eq:est} over time
(c.f. Remark~\ref{rem:num}).  \new{Nevertheless, taking the delays and
  sampling into account while designing the controller using the
  predictor-based scheme~\eqref{eq:u2} significantly increases the
  robustness of the closed-loop system relative to a design that is
  oblivious to delays and sampling. As shown in~\citep{FM-MM-TND:13},
  the asymptotic stability of the latter can only be guaranteed for
  this example \emph{without actuation delays and event-triggering} if
  $\Delta_\tau + D_\psi \le 7.1 \times 10^{-3}$ (given that, using the
  notation therein, we have $c_1 = 25, c_2 = 29/9, c_3 = 772$), which
  is more than two orders of magnitude more conservative than the
  empirical bound shown in Figure~\ref{fig:ex1}(d).}
  
\new{Finally, we have investigated the robustness of the closed-loop
  system to external disturbances (which are not theoretically
  included in our analysis but inevitably exist in practice). In an
  event-triggered system, disturbances may lead to instability and/or
  Zeno behavior, cf.~\citep{VSD-DPB-WPMHH:17}. However, as shown in
  Figure~\ref{fig:ex1}(e-f), neither instability nor Zeno behavior
  occurs when adding (any strength of) the disturbance here,
  highlighting the practical relevance of the proposed event-triggered
  scheme.}  }\oprocend
\end{example}

\begin{figure*}
  \centering
\subfloat[]{\small 
\includegraphics[width=0.32\linewidth]{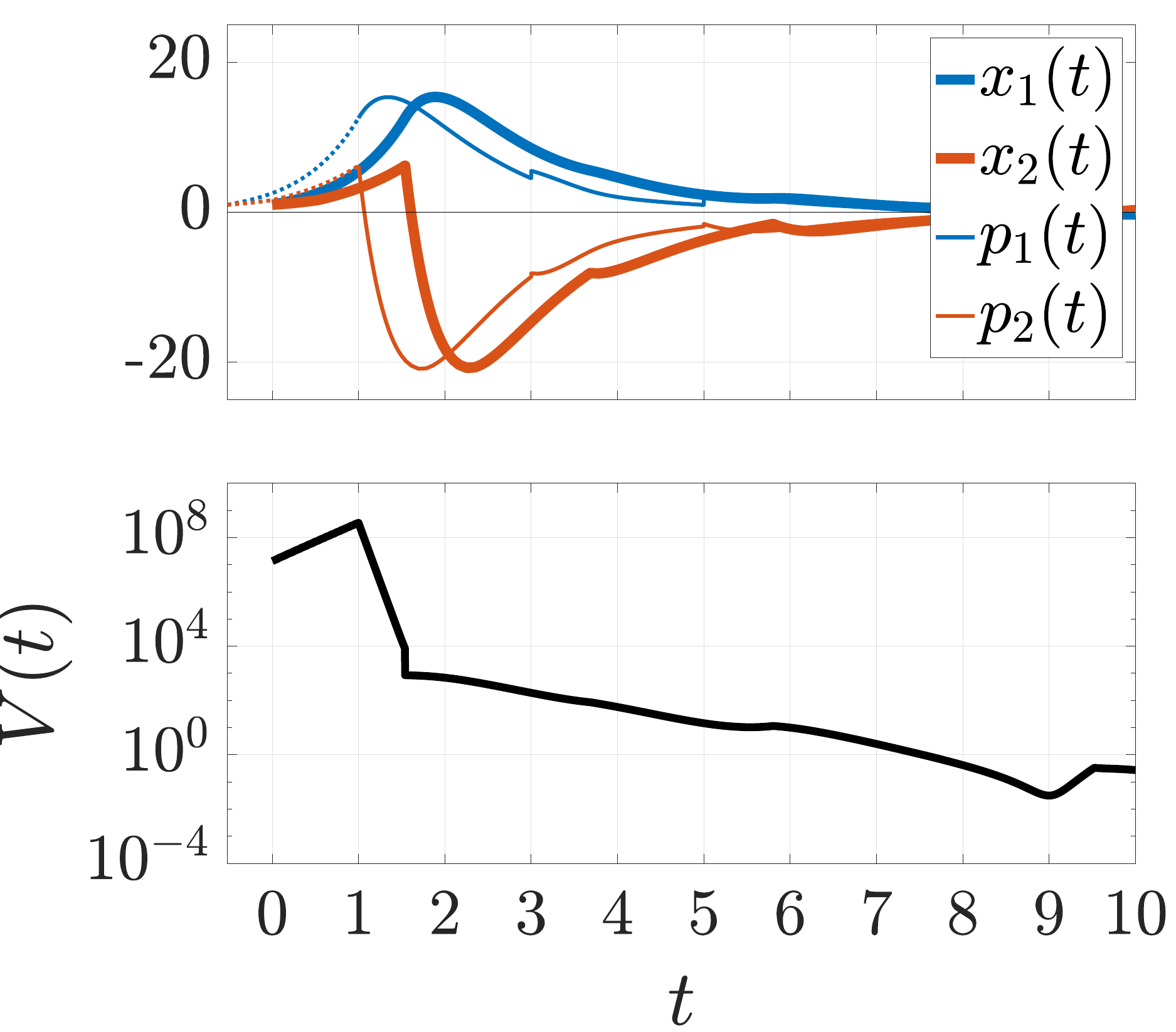}
}
\hspace{10pt}
\subfloat[]{\small 
\includegraphics[width=0.305\linewidth]{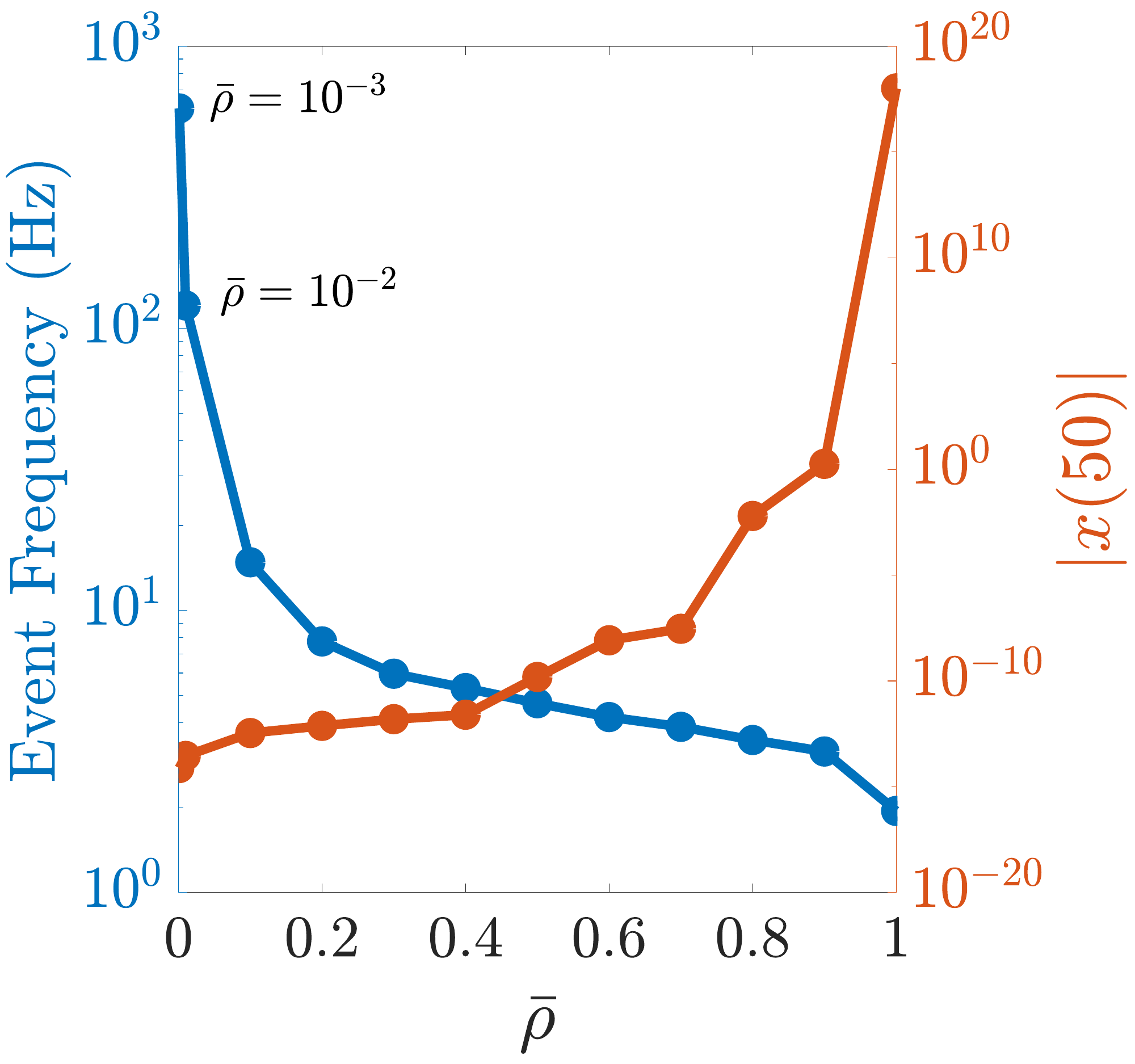}
}
\hspace{3pt}
\subfloat[]{\small 
\includegraphics[width=0.29\linewidth]{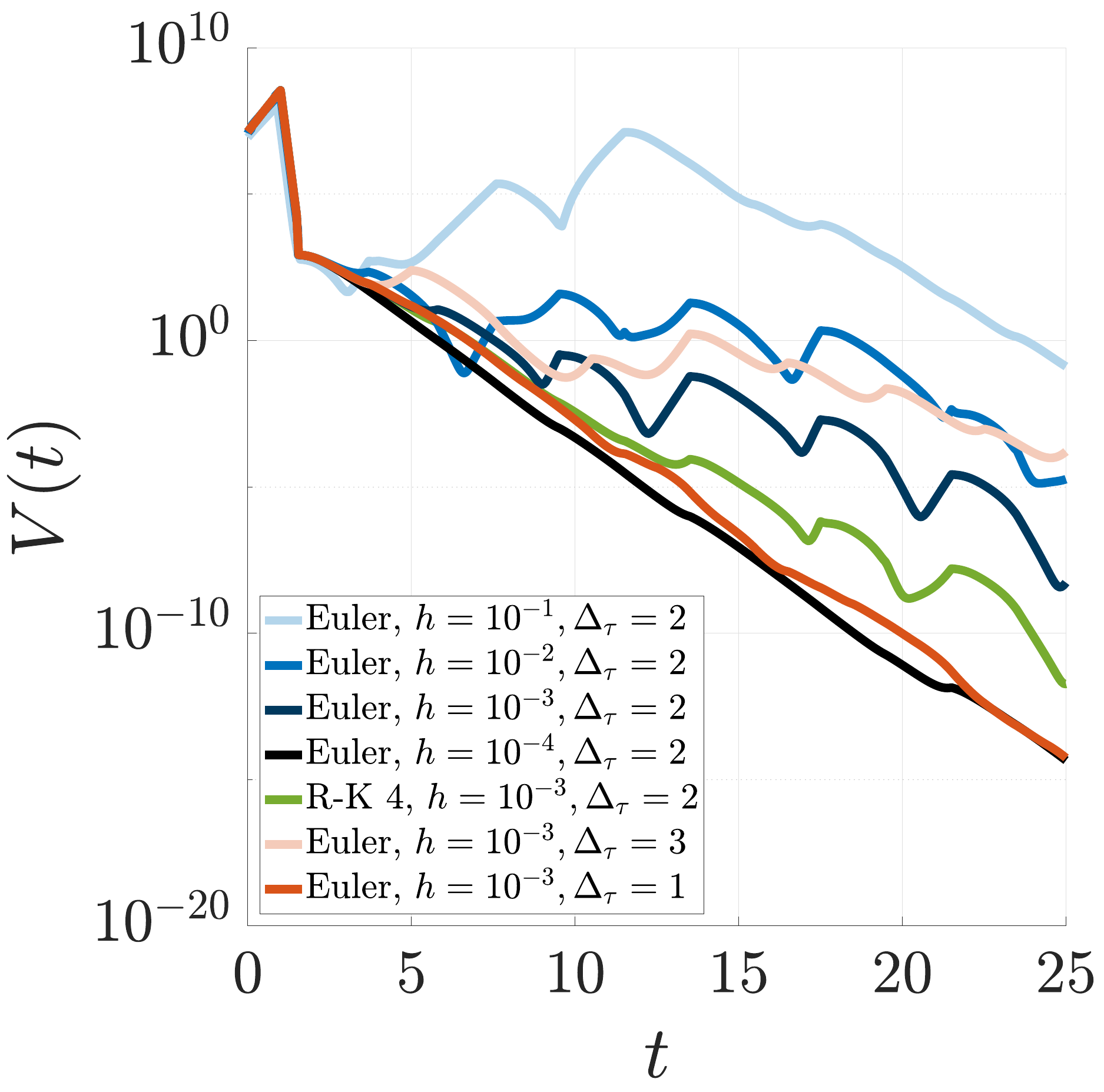}
}
\\
\subfloat[]{\small 
\includegraphics[width=0.365\linewidth]{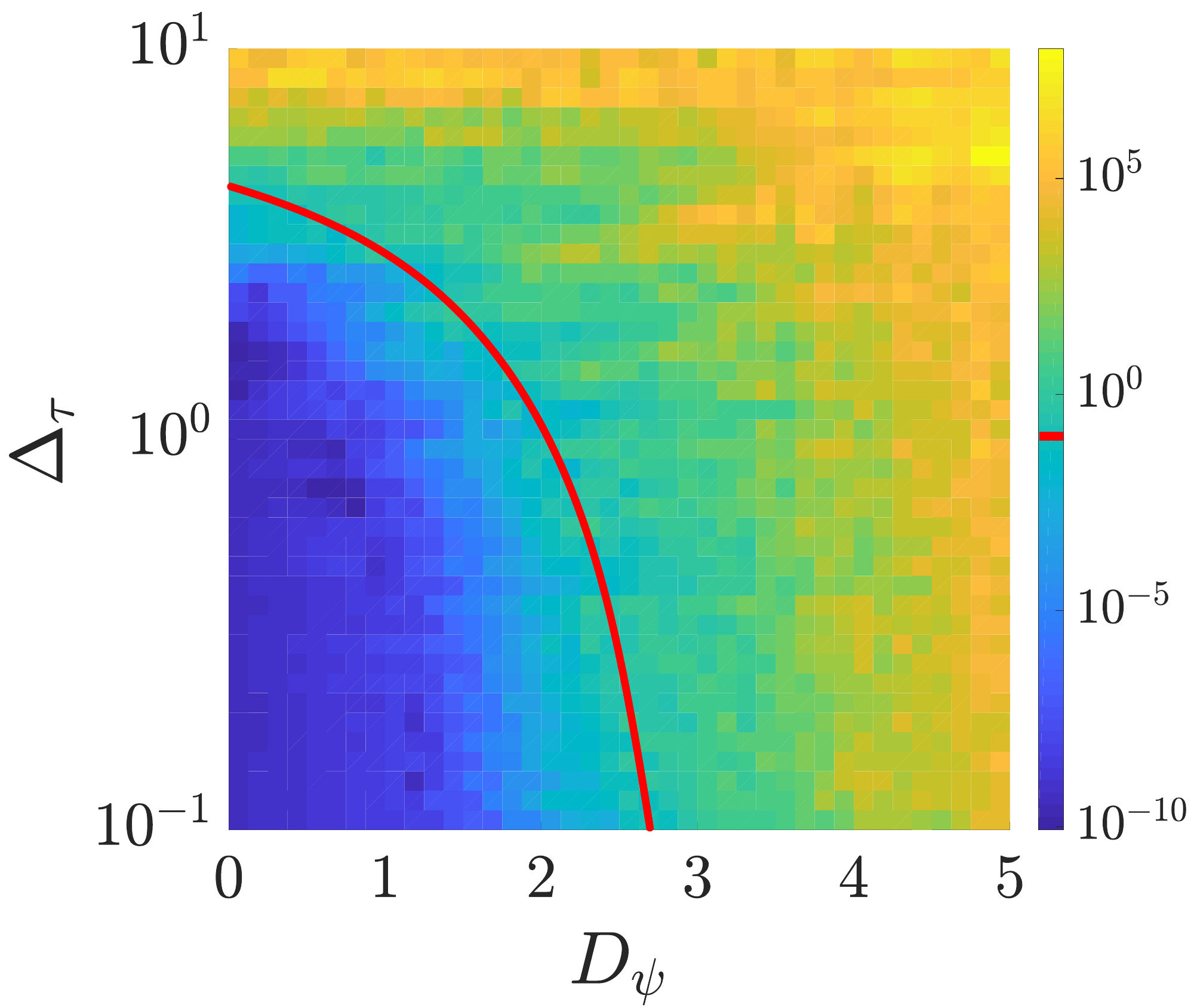}
}
\subfloat[]{\small 
\includegraphics[width=0.305\linewidth]{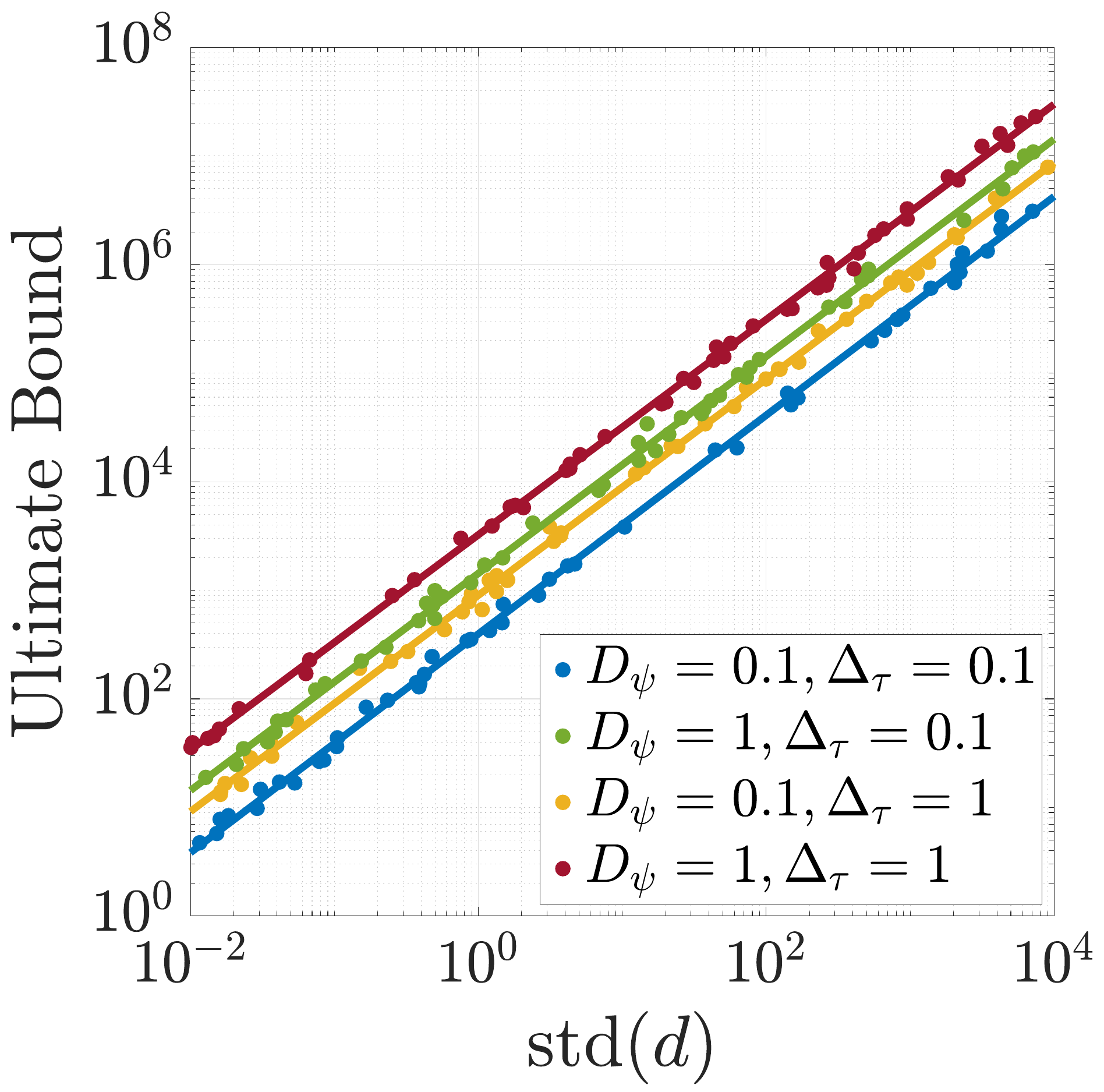}
}
\subfloat[]{\small 
\includegraphics[width=0.305\linewidth]{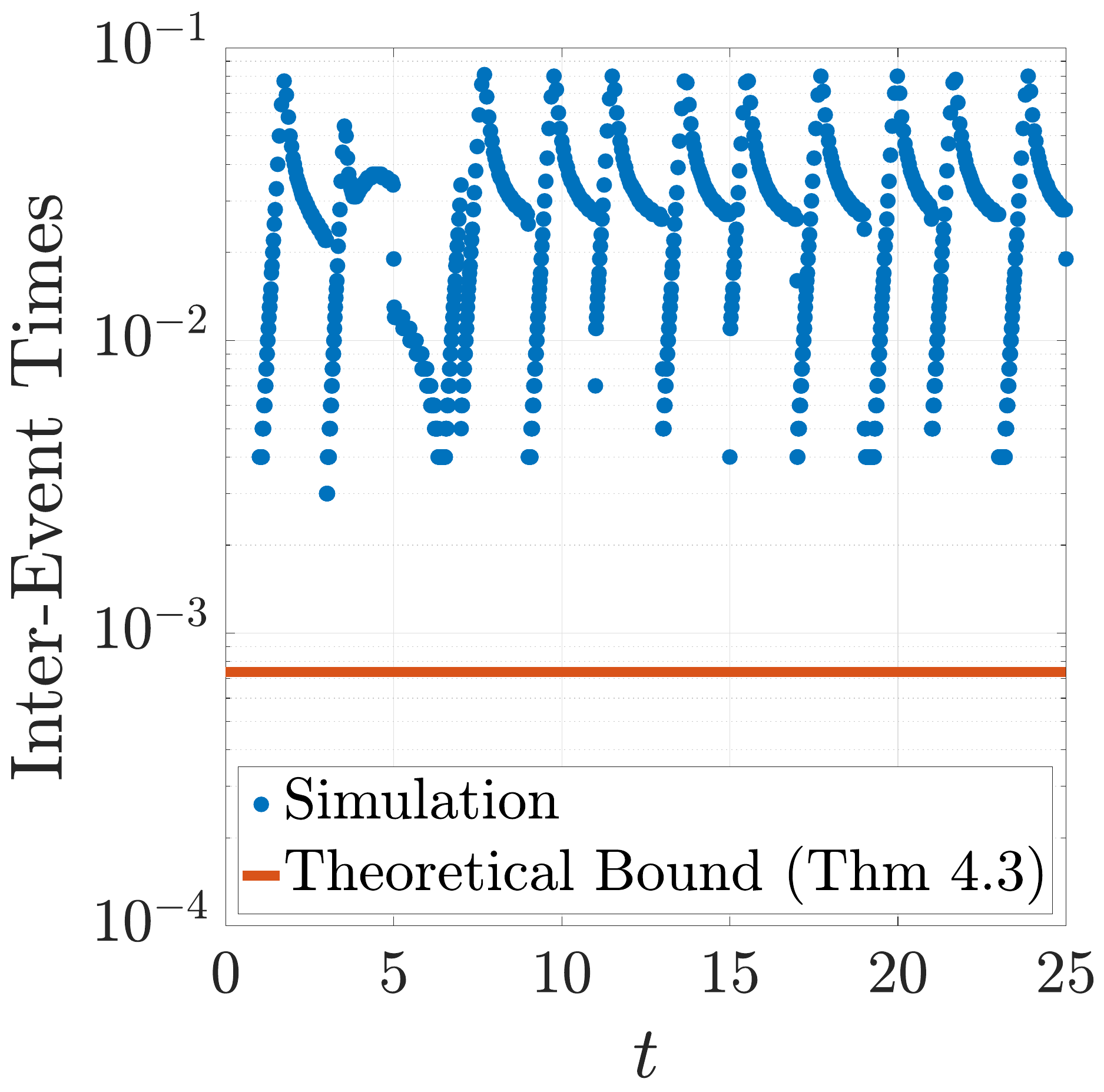}
}
\caption{\new{Simulation results for  Example~\ref{ex:1}. Unless otherwise stated, we use $x(0) = (1, 1)$, $\theta = 0.5$, $b = 10$,  $\Delta_\tau = 2$, $D_\psi = 1$, and Euler discretization with $h = 10^{-3}$. %
\textbf{(a)} Sample trajectories.  
The dotted portion of $p(t)$ corresponds to the times $[\phi(0),
\psi^{-1}(0))$ and is plotted only for illustration purposes (not used
by the controller). %
\textbf{(b)} The event-frequency and average of $|x(50)|$ over 100
random initial conditions as a function of $\bar \rho$. %
\textbf{(c)} The effect of discretization and state sampling on
stability. While stepsize $h$ and sampling rate $1/\Delta_\tau$ have a
strong impact on stability (blue and red curves, resp.), the effect of
discretization order is less significant (green curve, 4th order
Runge-Kutta). %
\textbf{(d)} Heat map of the average of $|x(25)|$ over $10$ random
initial conditions drawn from standard normal distribution. %
The red line shows an approximate border of stability.  \textbf{(e-f)}
Numerical verification of the robustness of the event-triggered
controller to additive disturbances: we augment~\eqref{eq:dynamics} as
$\dot x = f(x, u_p) + d$, where $d$ is zero-mean, white, and
Gaussian. %
\textbf{(e)} The estimate of the ultimate bound of state ($\max_{i =
  1, 2} \limsup_{t \to \infty} |x_i(t)|$) for varying standard
deviation of $d$. The value of the ultimate bound depends on sampling
delay and frequency, but the state always remains bounded for bounded
disturbances and the best linear fit always has a slope $\simeq 1$, a
behavior akin to globally input-to-state stable linear systems. %
\textbf{(f)} The inter-event times $\{t_{k + 1} - t_k\}_{k \ge 0}$ for
std($d$) = 1. Unlike~\citep{DPB-WPMHH:14}, the minimum inter-event
time is lower bounded by $\delta$ in Theorem~\ref{thm:zeno&gas}
irrespective of the existence or strength of disturbance (as long as
$\Delta_\tau > \delta$) due to the fact that sensing only occurs at
discrete-time instances $\{\tau_\ell\}$, making the controller
oblivious to disturbance over each $\Delta_\tau$ period. This may in
principle lead to instability ($|x| \to \infty$) but we see from (e)
that this is not the case. %
}}
    \label{fig:ex1}
    \vspace*{-1.5ex}
\end{figure*}

\begin{example}\longthmtitle{Non-compliant Nonlinear
    System}\label{ex:2}
  \rm Here, we consider an example that violates several of our
  assumptions.
  Let
  \new{
    \begin{align*}
      &f(x, u) = (A + \Delta A) x + B u + E x_1^3, \quad E = [0 \ \
      1]^T,
      \\
      &t - \phi(t) = D + a \sin(t), \ \tau_\ell = \ell
      \Delta_\tau, 
      \ \psi(t) = t - \frac{1 -
        e^{-t}}{2}, 
    \end{align*}
    where $A$ and $B$ are as in Example~\ref{ex:1}. The nominal delay
    $D$ and nominal coefficient matrix $A$ are known but their
    perturbations $a \sin(t)$ and $\Delta A$ are not (the controller
    \emph{assumes} $\phi(t) = t - D$ and $f(x, u) = Ax + Bu +
    Ex_1^3$).  We generate the elements of $\Delta A$ independently
    from $\Nc(0, \sigma_A^2)$.  Furthermore, in our simulation, the
    actual time that it takes for a sensor message $x(\tau_\ell)$ to
    reach the controller is \emph{not} the nominal delay
    $\psi^{-1}(\tau_\ell) - \tau_\ell$ but a random variable
    $D^\psi_\ell$, where
    \begin{align*}
      E[D^\psi_\ell] = \psi^{-1}(\tau_\ell) - \tau_\ell, \quad
      \var(D^\psi_\ell) = \sigma_\psi > 0.
    \end{align*}
    This serves to illustrate how the delay function $\psi$ (and
    similarly $\phi$), though being continuous and deterministic in
    our treatment, can be used to compensate for (in addition to
    physical sensor lag) computation and communication delays that are
    discrete and stochastic in nature%
    \footnote{Since the triggering times $\tau_\ell$ are themselves
      random and vary from execution to execution, the function $\psi$
      is defined for all $t$ even though only the discrete sequence
      $\{\psi^{-1}(\tau_\ell)\}$ is relevant for each execution.}.  }

  Moreover, 
  $K(x) = -6x_1 - 5 x_2 - x_1^3$ makes the closed-loop system ISS but
  is not globally Lipschitz, and the zero-input system exhibits finite
  escape time.
  The simulation results of this example are illustrated in
  Figure~\ref{fig:ex2}. It can be seen that although $V$ is
  significantly non-monotonic, the event-triggered controller is able
  to stabilize the system. While a thorough investigation of the
  stability of the resulting stochastic dynamical system reaches far
  beyond our theoretical guarantees, this example suggests that the
  proposed controller is robust to small violations of its assumptions
  and is thus applicable to a wider class of systems than those
  satisfying Assumption~\ref{assum}.  \oprocend
\end{example}

\begin{figure}
  \centering
\includegraphics[width=\linewidth]{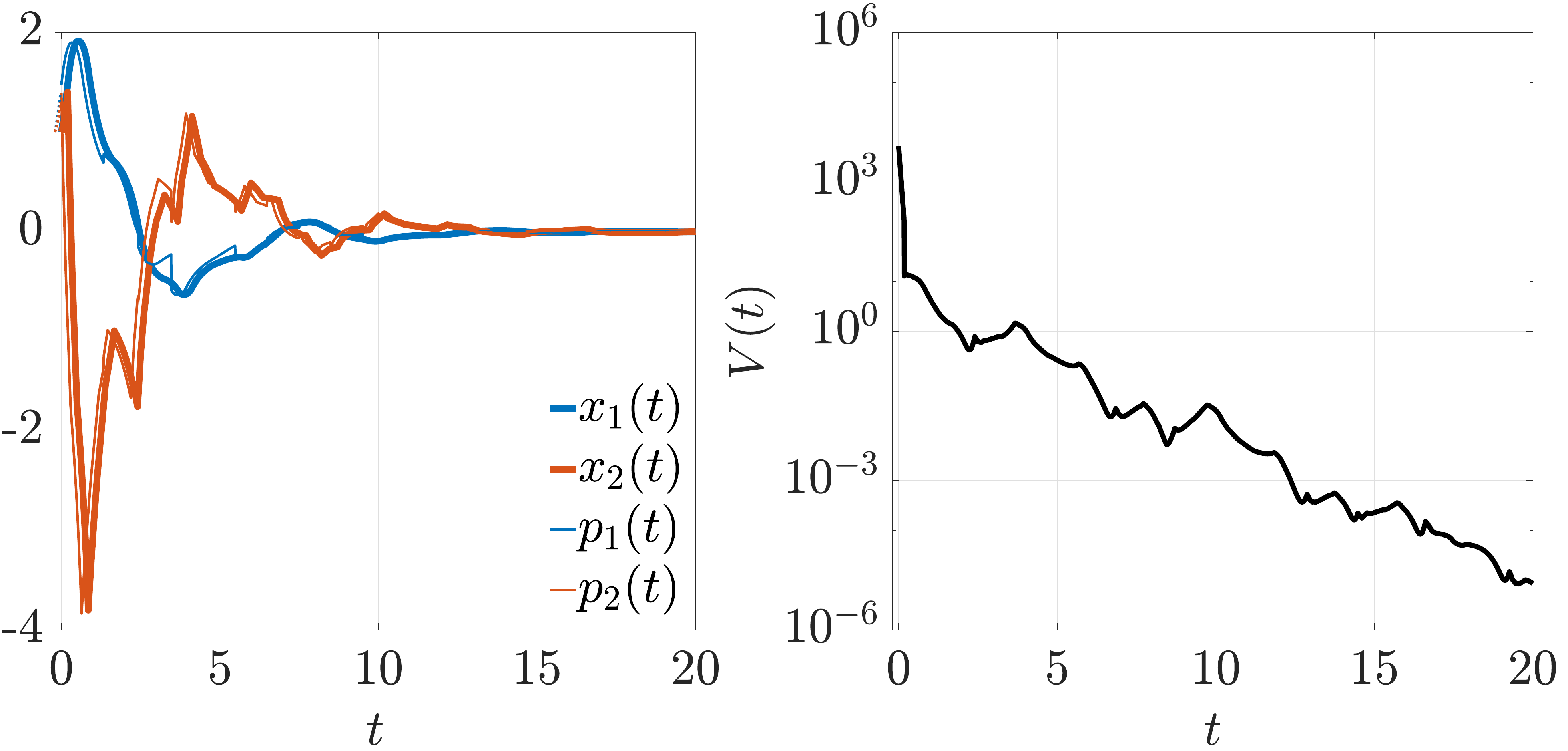}\ 
\caption{\new{Simulation of the non-compliant system in
    Example~\ref{ex:2}. We have used $x(0) = (1, 1)$, $\theta = 0.5$,
    $b = 10$, $a = 0.01$, $D = 0.2$, $\Delta_\tau = 1$, $\mu_\psi =
    0.1$, $\sigma_\psi = \sigma_A = 0.02$, triggering condition
    $|e(t)| \le 0.5 |p(t)|$, and Euler discretization of the
    continuous-time dynamics with $h = 10^{-2}$.}}
  \label{fig:ex2}
\end{figure}

\section{Conclusions and Future Work}

We have proposed a prediction-based event-triggered control scheme for
the stabilization of nonlinear systems with sensing and actuation
delays.  Assuming known time delay, globally-Lipschitz input-to-state
stabilizability, and state feedback, we have shown that the
closed-loop system is globally asymptotically stable and the
inter-event times are uniformly lower bounded. We have specialized our
results for linear systems, providing explicit expressions for our
design and analysis steps, and further studied the
sampling-convergence trade-off characteristic of event-triggered
strategies. Finally, we have addressed the numerical challenges that
arise in the computation of predictor feedback and demonstrated the
effectiveness of our approach in simulation. Regarding future work, we
highlight the extension of our results to systems with disturbances,
unknown input delays, or output feedback, the characterization of the
robustness properties resulting from incorporating the most recently
available state information, the relaxation of the global Lipschitz
requirement on the input-to-state stabilizer, and the study of the
effect on performance of the numerical implementation of the
event-triggered controller.

\section*{Acknowledgments}
This work was supported by NSF Award CNS-1446891 (EN, PT, and JC) and
AFOSR Award FA9550-15-1-0108 (JC).

\end{document}